\documentclass[12pt]{article}
\usepackage{amssymb,amsmath,ntheorem,amsfonts,mathtools,eurosym,geometry,graphicx,caption,color,setspace,sectsty,comment,caption,natbib,pdflscape,array,bm,tabularx,tikz,float,fp,dsfont}
\usepackage[normalem]{ulem}
\usepackage[figuresleft]{rotating}
\usepackage[multiple, bottom]{footmisc}
\usepackage{dcolumn}
\usepackage{makecell}
\usepackage{ragged2e} 
\usepackage[hidelinks]{hyperref} 
\usepackage{subcaption}
\usepackage{adjustbox}
\usepackage{algorithm, algpseudocode}
\usepackage{pifont}
\newcommand{\cmark}{\ding{51}}%
\newcommand{\xmark}{\ding{55}}%

\newcommand{\tablenote}[1]{\newline \vspace{0mm} \newline \footnotesize \justify \textbf{Note: }{#1}} 
\newcommand{\ind}{\perp\!\!\!\perp} 

\theoremstyle{break}
\newtheorem{theorem}{Theorem}
\newtheorem{lemma}{Lemma}
\newtheorem{assumption}{Assumption}

\newtheorem{remark}{Remark}
\newenvironment{proof}[1][Proof]{\noindent\textbf{#1:} }{\ \rule{0.5em}{0.5em}}

\newcommand\numberthis{\addtocounter{equation}{1}\tag{\theequation}}

\newcolumntype{L}[1]{>{\raggedright\let\newline\\arraybackslash\hspace{0pt}}m{#1}}
\newcolumntype{C}[1]{>{\centering\let\newline\\arraybackslash\hspace{0pt}}m{#1}}
\newcolumntype{R}[1]{>{\raggedleft\let\newline\\arraybackslash\hspace{0pt}}m{#1}}
\newcolumntype{Y}{>{\centering\arraybackslash}X}

\begin{document}

\begin{titlepage}
\title{Residualised Treatment Intensity and the Estimation of Average Partial Effects}
\author{Julius Schäper\thanks{Department of Economics, University of Zurich. I would like to thank Victor Chernozhukov, Juan Carlos Escanciano, Bo Honoré, Michal Kolesár, Damian Kozbur, Filippo Palomba, Jonathan Roth, Rainer Winkelmann, Mark Watson, Michael Wolf, Kaspar Wüthrich and the participants at the PhD Seminars at Princeton University and University of Zurich for their helpful comments and suggestions.}}
\date{\today}
\maketitle
\begin{abstract}
\noindent 
This paper introduces R-OLS, an estimator for the average partial effect (APE) of a continuous treatment variable on an outcome variable in the presence of non-linear and non-additively separable confounding of unknown form. Identification of the APE is achieved by generalising Stein's Lemma \citep{stein1981estimation}, leveraging an exogenous error component in the treatment along with a flexible functional relationship between the treatment and the confounders. The identification results for R-OLS are used to characterize the properties of Double/Debiased Machine Learning \citep{Chernozhukov_et_al_2018_DebiasedML}, specifying the conditions under which the APE is estimated consistently. A novel decomposition of the ordinary least squares estimand provides intuition for these results. Monte Carlo simulations demonstrate that the proposed estimator outperforms existing methods, delivering accurate estimates of the true APE and exhibiting robustness to moderate violations of its underlying assumptions. The methodology is further illustrated through an empirical application to \cite{fetzer2019_austerity_brexit}.
\\
\vspace{0in}\\
\noindent\textbf{Keywords:} \emph{average partial effect, distributional moments, residualisation, identification, ordinary least squares, debiased machine learning} \\

\bigskip
\end{abstract}
\setcounter{page}{0}
\thispagestyle{empty}
\end{titlepage}
\pagebreak \newpage

\section{Introduction} \label{sec:introduction}
This paper is concerned with identification and estimation of the average partial effect $E[\partial Y/\partial X]$ for the model $Y = g(X,Z) + \varepsilon$, where $X$ is a continuous treatment variable of interest and $Z$ is a set of confounders. A novel estimator, R-OLS, is presented, which exploits an exogenous error component of the treatment to estimate the average partial effect (APE). Notably, R-OLS is suitable to estimate treatment effects in models characterized by non-linear and non-additively separable confounding of unknown form.

The estimator adds to two strands of the literature. The first one aims to identify and estimate the average partial effect within the aforementioned model by imposing restrictions on $g(X, Z)$. For instance, the canonical linear regression model assumes linear and additively separable confounding and accounts for confounders by adding them as controls. Building upon this, \cite{Robinson_1988} extends the framework to the partially linear model $Y = \beta X + \theta(Z) + \varepsilon$, accommodating additively separable but potentially non-linear confounding. Taking a further step, \cite{Graham_Pinto_2022} incorporate heterogeneous treatment effects by allowing for interactions between $X$ and $Z$ in the Y-DGP, but maintain the linearity of $X$. Without linearity in $X$, their estimator recovers a weighted average of derivatives of Y wrt. X with unknown weights. 

A second strand of the literature achieves identification of the average partial effect by placing restrictions on the joint distribution of $X$ and $Z$. In a setting without confounding, Stein's Lemma \citep{stein1981estimation, Ross_Steins_Method} shows that a linear regression of $Y$ on $X$ identifies the average partial effect if and only if $X$ is normally distributed. Under confounding, joint normality of $X$ and $Z$ enables the identification of the APE \citep{Stoker_1986,Powell_et_al_1989}. \cite{LANDSMAN2008} extend Stein's Lemma to elliptical distributions, but only estimate the APE for a transformed variable $X^*$, resulting in a reweighting of the original APE. Considering endogeneity and unobservables, \cite{Cuesta_Steins_Lemma} show that joint normality of the treatment and an instrumental variable enable the IV-Estimator to estimate the APE under additive separability of unobservables.

The R-OLS estimator combines elements of the previous two approaches and leverages the trade-off between the complexity of the $Y$-DGP and distributional assumptions on the covariates. Through the lens of the first strand of literature, restricting $g(.)$, it enables a more flexible functional form for $g(.)$, at the cost of making assumptions on the exogenous variation in $X$. Conversely, viewed within the context of the second strand of literature R-OLS relaxes the assumption of joint normality of covariates considerably, but restricts $g(.)$. Therefore, the R-OLS estimator can be seen as a combination of both previous approaches, allowing more flexibility in certain assumptions while imposing constraints on others. Furthermore, due to the structure of the assumptions made, R-OLS allows for customisation of the trade-off to the empirical application, since the flexibility in the $Y$-DGP is directly linked to the extend of the distributional assumptions imposed. 

Concretely, the average partial effect of treatment $X$ on $Y$, $E\left[ \partial_{X_i} Y_i \right]$, is identified by a bivariate regression of $Y_i = g(X_i,Z_i) + \varepsilon_i$ on the exogenous variation $\nu_{i}$ in $X_i = r(Z_i) + \nu_{i}$. Here, $r(Z_i)$ can be any function describing the confounding between the treatment $X_i$ and covariates $Z_i$. The Y-DGP is characterised by $g(X_i, Z_i)$ and assumed to be an interaction of polynomials of the treatment and arbitrary functions of the confounders. 

Under these assumptions, conditions on the distributional moments of $\nu$ establish an equivalence between R-OLS and the average partial effect of the treatment. The conditions on the distributional moments are directly related to the order of polynomials in the outcome DGP and hence enable a flexible trade-off between both assumptions. On one side of the trade-off, the distributional assumptions on $\nu$ only impose a mean of zero, but the $Y$-DGP is required to be linear in $X$. On the other end, the exogenous error $\nu$ unconditionally follows a normal distribution with mean zero, but the average partial effect is identified without placing assumptions on the DGP of $Y$. Specifically, if one is willing to assume a $Y$-DGP quadratic in $X$, then the exogeneous error just needs to have a first and third moment of zero to guarantee that R-OLS estimates the APE under arbitrarily complex confounding.

Similar to \cite{Robins_et_al_1992}'s E-estimator and \cite{MJ_Lee_2018}'s propensity score residual estimator, R-OLS is particularly useful under low-complexity treatment DGPs where the prediction of the error component is feasible or the DGP known. However, compared to earlier work, this paper goes one step further and shows that the average partial effect of a continuous treatment is identified by R-OLS, even under a more flexible semi-parametric model for the outcome, if certain distributional assumptions hold. 

Following these identification results, it is shown how the average partial effect can be estimated under knowledge of the treatment DGP using the Frisch-Waugh-Lovell theorem. The asymptotic distribution of the estimator is derived and an estimator of the variance provided. 
Afterwards, Double/Debiased Machine Learning (DML) \citep{Belloni_et_al_Lasso_2014, Chernozhukov_et_al_2018_DebiasedML, chernozhukov2022autodml} is leveraged to enable statistical inference for the APE in settings without knowledge of the treatment DGP by showing that, under the assumptions of R-OLS, the Neyman-orthogonal moment for the partially linear model estimates the APE.

Intuition for R-OLS is provided via a novel decomposition of OLS and IV estimands, showing that both estimands approximate the true outcome DGP with a Taylor expansion, take the derivative of said polynomials and apply weights to each polynomial. However, these weights do not equal unity and hence distort estimates away from the APE. The assumptions made in this paper guarantee that these weights equal unity and hence enable the estimation of the APE.

Simulation experiments demonstrate that R-OLS accurately estimates the APE even if the samples are small and/or the $Y$-DGP highly complex, improving on existing semi-parametric approaches. Estimates remain close to the true APE even under moderate violations of assumptions.

A real-world application re-estimating the effect of austerity on the UK Brexit referendum in \cite{fetzer2019_austerity_brexit} is provided. The average partial effect of austerity on UKIP voteshare in local elections confirms the original finding, however, no statistically significant effect on European elections is detected.

The remainder of the paper is structered as follows: Section \ref{sec:assump_estimator} formalizes the set-up and states the main result. Section \ref{sec:Estimation} discusses the feasibility and asymptotic distribution of the estimator. Section \ref{sec:weights} decomposes the weights underlying R-OLS estimation. Simulation results for a variety of data generating processes are presented in Section \ref{sec:Simulations}, and real-world applications are provided in Section \ref{sec:emp_illustrations}. Finally, Section \ref{sec:conclusion} concludes.

\section{Model Assumptions and Estimand} \label{sec:assump_estimator}
The general estimation procedure underlying R-OLS is as follows: Assume $X$ is determined by a function of covariates with arbitrary dependencies and some additive exogeneous variation that is related to $Y$ only through $X$. Since changes in the exogenous variation only influence $X$, leaving everything else constant, regressing $Y$ on this exogenous variation identifies the a causal effect of $X$ on $Y$. This section shows that the average partial effect can be estimated this way. 

\subsection{Assumptions}
We begin by describing the data generating process for which we want to estimate the average partial effect of $X$ on $Y$.

\begin{assumption}[Data Generating Process]
    \label{assump_DGP} 
    For independently and identically distributed random vectors $(Y_i, X_i, Z_i)$ it holds that 
    \begin{align}
        &Y_i = \sum_{m=0}^M X_{i}^mg_m(Z_{i}) + \varepsilon_i \quad \text{for} \quad M \in \mathbb{N} \label{DGP:Y} \\
        &X_{i} = r\left(Z_{i}\right)+\nu_{i} \label{DGP:X}
    \end{align}
    with $X_i \in \mathbb{R}$, $Z_i \in \mathbb{R}^K$, $E[\varepsilon_i \mid X_i, Z_i] = 0$, $\nu_{i} \ind Z_{i}$, $E[|\nu_i^2|] < \infty$ and $E[|\nu_iY_i|] < \infty$. Furthermore, $g_m: \mathbb{R}^{K} \to \mathbb{R}$ and $r: \mathbb{R}^{K} \to \mathbb{R}$.
\end{assumption}

Assumption \ref{assump_DGP} defines the data generating process of $Y_i$ to be a polynomial of the independent variable of interest, $X_i$, multiplied by an arbitrarily complex function $g_m(Z_i)$ of all the other independent variables $Z_i$, where $Z_i$ can differ between equations (\ref{DGP:Y}) and (\ref{DGP:X}). Although seemingly restrictive, the Y-DGP specified in Assumption \ref{assump_DGP}, is quite general. Firstly, $g_m(Z_{i})$ can be any function and differ with $m$. Secondly, by the Weierstrass Approximation Theorem every continuous function defined on a closed interval $[a, b]$ can be uniformly approximated as closely as desired by a polynomial function. Therefore, the functional form of the Y-DGP in Assumption \ref{assump_DGP} allows for arbitrary flexibility in the outcome data generating process, if $M\to \infty$.

The variable of interest, $X_i$, is assumed to be determined by an unrestricted function $r(Z_i)$ of all other independent variables plus an exogenous error $\nu_i$.\footnote{Throughout the paper $\nu_i$ is referred to as exogenous error, error or exogenous variation. The estimate $\hat{\nu_i}$ will be referred to by residual variation or residual. In the machine learning literature $\nu_i$ is referred to as the irreducible error and $\hat{\nu_i}$ is called the prediction error.} 
The estimand presented in this paper does not require the functional form of $r(Z_i)$ to be known, but rather assumes that $\nu_i$ is known. This has implications for settings in which $X_i$ can be manipulated or its exogenous variation is known. When describing our estimation procedure, we will leverage orthogonalised moment conditions and estimate $\nu_i$ with machine learning.

The exogenous error $\nu_i$ is the main ingredient in the estimation procedure outlined in Theorem \ref{thm:main} and enables the estimation of the average partial effect via a univariate regression under the following assumption:

\begin{assumption}[Distributional Moments of the Error]
    The exogenous error $\nu_i$ satisfies the distributional moments
    \label{assump_errors}
    \begin{align*}
       \frac{1}{(p+1) E\left[\nu_{i}^2\right]}  E\left[\nu_{i}^{p+2}\right] &= E\left[\nu_{i}^{p}\right] \quad \text{for} \quad p \in \mathbb{N} \quad \text{and} \quad 0 \leq p \leq M-1 \\
       E\left[v_i\right]&=0,
    \end{align*}
    where $M$ is determined by the highest order polynomial of $X$ in equation (\ref{DGP:Y}).
\end{assumption}

Assumption \ref{assump_errors} implies that the distribution of the exogenous variation influencing the variable $X$ has odd moments that are equal to zero, while even moments conform to a specific series. 

Following Lemma \ref{lemma:normal_moments} in the Appendix, Assumption \ref{assump_errors} is fulfilled for normally distributed errors up to $M=\infty$, regardless of the variance, as long as the expected value is zero. However, it is important to note that it is not necessary for the errors to satisfy these moments up to an infinite order; it suffices to adhere to them up to order $M+1$, where $M$ represents the maximum order of the polynomial terms of $X$ in equation (\ref{DGP:Y}). For example, in the partial linear model of \cite{Robinson_1988} with $M=1$, Assumption \ref{assump_errors} requires that the errors have a mean of zero.

\subsection{Estimand}
Using Assumptions \ref{assump_DGP} and \ref{assump_errors}, we can define the R-OLS estimand for the average partial effect:

\begin{theorem}[R-OLS]
    \label{thm:main}
    Under Assumptions \ref{assump_DGP} and \ref{assump_errors}, the R-OLS estimand:
    \begin{align}
      \beta = \frac{E[\nu_{i}Y_i]}{E[\nu_{i}^2]} 
    \end{align}
    is equivalent to the average partial effect of $X_i$ on $Y_i$:
    \begin{align}
      E_{X,Z, \varepsilon}\left( \partial_{X_i} Y_i\right) 
    \end{align}
\end{theorem}
\vspace*{5mm}

For a proof of Theorem \ref{thm:main} see Appendix A. The proof relies heavily on combining the functional form for the Y-DGP in Assumption \ref{assump_DGP} with the moments of the exogenous error in Assumption \ref{assump_errors}. Together, the assumptions allow matching R-OLS with the average partial effect after careful expansion of the polynomial structure.  

Since the Y-DGP in Assumption \ref{assump_DGP} is highly flexible and can approximate any functional form, if $M\to \infty$, we can conclude that the R-OLS estimand identifies the average partial effect without assumptions on the Y-DGP, as long as $\nu_i$ follows a normal distribution with expected value of zero. For finite orders of $M$, the distributional moments of $\nu_i$ need to match the respective moments of the normal distribution up to order $M+1$.

Theorem \ref{thm:main} generalises Stein's Lemma \citep{stein1981estimation} by leveraging a trade-off between functional form and distributional assumptions. It shows that normality of $X_i$ can be relaxed to the structure of moments in Assumption \ref{assump_DGP} depending on the flexibility of the outcome model. 

\begin{remark}[Independence of $\nu_i$]
    Full independence of $\nu_i$ is not necessary. Only the  moments of $\nu_i$ up to $M+1$ need to be independent of $Z_i$. However, even this assumption can be relaxed to mean independence, under an alternative version of Assumption \ref{assump_errors}:
    \begin{align*}
        \frac{1}{(p+1) E\left[\nu_{i}^2\right]}  E\left[\nu_{i}^{p+2} \mid Z_i \right] &= E\left[\nu_{i}^{p} \mid Z_i \right] \quad \text{for} \quad p \in \mathbb{N} \quad \text{and} \quad 0 \leq p \leq M-1 \\
        E\left[v_i \mid Z_i \right]&=0.
    \end{align*}
    This means $\nu_i$ needs to satisfy the moments of the normal distribution up to order $M+1$ conditional on $Z_i$, if we also assume $E\left[\nu_{i}^2\right] = E\left[\nu_{i}^2 \mid Z_i \right]$. Given the difficulty of imagining a $\nu$-DGP that satisfies these conditions, but not independence, we assume independence.
\end{remark}

\subsection{Example: $M=2$}
To better understand the implications of Theorem \ref{thm:main}, consider the following example. For $M=2$, the DGPs in Assumption \ref{assump_DGP} can be written as
\begin{equation}
    \begin{gathered}
        \label{example_DGP}
        Y_i = g(Z_{i}) + X_{i}  g_1(Z_{i}) + X_{i}^2  g_2(Z_{i}) + \varepsilon_i   \\
        X_{i} =r\left(Z_{i}\right)+\nu_{i}
    \end{gathered}
\end{equation}
Theorem \ref{thm:main} implies that the R-OLS estimand:
\begin{align*}
    \beta = \frac{E[\nu_{i}Y_i]}{E[\nu_{i}^2]} 
\end{align*}
is equivalent to the average partial effect of $X_i$ on $Y_i$ if the distribution of $\nu_i$ satisfies the following two moment conditions:
\begin{align*}
    0 &= E\left[ \nu_{i} \right] = E\left[ \nu_{i}^{3} \right]
\end{align*}
These conditions hold for the normal distribution or any other distribution that is symmetric around zero. If a cubic term of $X$ was to be added to the DGPs in equation (\ref{example_DGP}), the errors would need to satisfy $Kurtosis(\nu_i)=\frac{E\left[ \nu_{i}^{4} \right]}{E^2\left[ \nu_{i}^2 \right]} = 3$, a property defining the class of mesokurtotic distributions\footnote{Examples of mesokurtotic distributions can be found in Appendix C.}. With a fourth degree polynomial of $X_i$, the errors would also need to fulfill $E\left[ \nu_{i}^{5} \right] = 0$. In general, the higher the order of the $X_i$-polynomial in the $Y$-DGP, the more moments of the error need to satisfy Assumption \ref{assump_errors} and the more the error distribution tends to a normal distribution.

\section{Estimation} \label{sec:Estimation}
The previous section showed how the average partial effect can be estimated with R-OLS under knowledge of $\nu_i$. In this section we will relax this assumption and replace $\nu_i$ with an estimate. We will begin by discussing estimation using OLS, assuming the functional form linking the treatment variable $X_i$ and the confounders $Z_i$ is known. Afterwards, we will consider estimation via machine learning without knowledge of the functional form. 

Regardless of which method is used to estimate $\nu_i$, the R-OLS estimator is given by:
\begin{align*}
    \hat{\beta} = \frac{\frac{1}{N} \sum_{i=1}^N \hat{\nu}_{i}Y_i}{\frac{1}{N} \sum_{i=1}^N \hat{\nu}_{i}^2} 
\end{align*}
Here, the population moments have been replaced with sample moments and $\nu_i$ with an estimate $\hat{\nu}_i$. Assuming the existence of a consistent estimator $\nu_i$, the weak law of large numbers can be applied to the numerator and denominator of the R-OLS estimator.\footnote{If $E[|\nu_i^2|] < \infty$ and $E[|\nu_iY_i|] < \infty$.} Then, by the continuous mapping theorem, $\hat{\beta} \xrightarrow{p} \beta$ as $n \to \infty$, proving the consistency of the R-OLS estimator for the APE following Theorem \ref{thm:main}. The R-OLS estimator itself is agnostic to the method used to predict $r(Z_i)$. The only assumption needed for the estimator to be consistent for the APE is that $\hat{\nu}_i \xrightarrow{p} \nu_i$ as $n \to \infty$. 
 
However, when estimating $\nu_i$ via machine learning, without knowledge of the functional form of $r(Z_i)$, even though the estimator might be consistent, the rate of convergence can be slower than $\sqrt{n}$, making inference difficult.\footnote{Under knowledge of the functional form in OLS, the convergence rate is $\sqrt{n}$ \citep{hansen2022econometrics}.}  To enable inference with machine learning estimators the problem will be recast within the Double/Debiased Machine Learning framework \citep{Chernozhukov_et_al_2018_DebiasedML} in Section \ref{sec:est_DML_ROLS}.

\subsection{Estimation with OLS} 
\label{sec:est_OLS}
We begin by assuming that the functional form of $r(Z_i)$ is known and allows for estimation via OLS. Under this assumption, $X_i$ can be modeled as $X_i \sim r(Z_i)$ and the resulting residuals can be utilised to estimate the APE with R-OLS. By the Frisch-Waugh-Lovell Theorem, this approach is equivalent to directly regressing $Y_i$ on $X_i$ and $r(Z_i)$, where in an abuse of notation $r(Z_i)$ is supposed to be the functional form of the X-DGP. This leads to:
\begin{align*}
    \hat{\beta} - E\left[ \partial_{X_i} Y_i\right] &= \left[(\Omega'\Omega)^{-1}(\Omega'Y)\right]_X - E\left[ \partial_{X_i} Y_i\right] \\
    &= \left[(\Omega'\Omega)^{-1}(\Omega'Y) - A \right]_X\\
    &= \left[(\Omega'\Omega)^{-1}(\Omega' (Y - \Omega A))\right]_X
\end{align*}
where the R-OLS estimator was rewritten in matrix form using the FWL-Theorem. $[.]_X$ symbolises the element of the vector/matrix corresponding to $X$. $\Omega$ is defined as $\{X, r(Z)\}$ and $A$ is the coefficient vector in the population level regression of $Y$ on $\Omega$. $[A]_X=E\left[ \partial_{X_i} Y_i\right]$ by the FWL-Theorem and Theorem \ref{thm:main}. This estimator has the asymptotic distribution 
\begin{align*}
    \sqrt{n}(\hat{\beta} - E\left[ \partial_{X_i} Y_i\right]) \xrightarrow{d} N(0, diag(\bm{V})_X) \quad \text{where} \\
    \bm{V} = (E[\Omega'\Omega])^{-1} E[\Omega\Omega'U^2] (E[\Omega'\Omega])^{-1} \quad \text{with} \quad U=Y - \Omega A
\end{align*}
under the assumptions that $E[Y^4] < \infty$, $E[|\Omega|^4] < \infty$ and $\Omega'\Omega$ positive definite. $\bm{V}$ can be estimated by replacing the population moments with their sample counterparts and $A$ with $\hat{A}$, the estimated coefficient vector from the regression of $Y$ on $\Omega$.  

Simplifications of the variance under homoscedasticity are not recommended since they implicitly impose strong assumptions on the $Y$-DGP. Similarly to a FWL-type regression, heteroscedasticity robust standard errors for the coefficients have to be estimated. This can easily be seen by plugging Assumption \ref{assump_DGP} into $E[\Omega\Omega'U^2]$:
\begin{align*}
    E[\Omega\Omega'U^2] &= E[\Omega\Omega' (Y - \Omega A)^2] \\
    &= E[\Omega\Omega' (\sum_{m=0}^M X^mg_m(Z) + \varepsilon - \Omega A)^2]
\end{align*}
Here, the expectation can not be split into $E[\Omega\Omega']$ and $E[(\sum_{m=0}^M X^mg_m(Z) + \varepsilon - \Omega A)^2]$ since $\sum_{m=0}^M X^mg_m(Z)$ and $\Omega A$ do not cancel out, creating dependence between both parts of the expectation and invalidating the use of homoscedastic standard errors.

\subsection{Naïve Estimation with Machine Learning}
\label{sec:est_naive_ML_ROLS}

If the functional form of $r(Z_i)$ is unknown, a machine learning estimator can be used to estimate it. For consistency of R-OLS, the estimation procedure must satisfy $r_n(Z_{i}) \xrightarrow{p} r(Z_{i})$ as $n \to \infty$. One algorithm satisfying these conditions is given by neural networks. Following \cite{HORNIK1989}, artifical neural networks are universal approximators that can approximate any Borel measurable function from one finite dimensional space to another to any desired degree of accuracy, provided the network is sufficiently large, the sample size goes to infinity and the function is sufficiently smooth. Following this, we outline a simple two-step procedure to estimate the average partial effect under the R-OLS framework when the data-generating process (DGP) for $X_i$ is unknown.

First, $r(Z_i)$ is estimated using a neutral network. With this estimate, the residuals $\hat{\nu}_i = X_i - \hat{r}(Z_i)$ are computed for each observation. In the second step, a standard OLS regression of $Y_i$ on $\hat{\nu}_i$ is performed. The coefficient on $\hat{\nu}_i$ from this regression provides the R-OLS estimate of the average partial effect. To construct valid confidence intervals for the R-OLS estimate, a bootstrap procedure is employed. This involves repeatedly resampling the data, re-estimating the APE using described procedure, and calculating confidence intervals of the R-OLS estimator based on the set of estimates.
    
It is however important to note that any method that consistently estimates $r(Z_i)$ can be used to estimate the average partial effect. The only requirement is that the method's predictions converge to $r(Z_i)$ with increasing sample size, residualising $X_i$ correctly.\footnote{Incorrect residualisation of $X_i$ is the origin of the biases described in \cite{goldsmithpinkham_et_al_2022} and \cite{winkelmann_2023}. Both papers show that by not modelling $X_{i}$'s DGP correctly, a correlation between the predicted residual, regressor and heterogeneous treatment effect can occur which results in a biased estimate.} 

Although the estimator might be consistent, it does not necessarily attain $\sqrt{n}$-consistency. As shown in equation (\ref{eq:bias_ROLS}), if the rate of convergence of $\hat{r}(Z_i)$ is slower than $\sqrt{n}$, the estimator will still converge to the true parameter, but at a rate slower than $\sqrt{n}$. To avoid asymptotic bias in this case, it is essential that the estimation errors are uncorrelated with $Y_i$. This issue can be addressed by combining the R-OLS identification result with the Double/Debiased Machine Learning framework of \cite{Chernozhukov_et_al_2018_DebiasedML}.

\subsection{Estimation with Double/Debiased Machine Learning}
\label{sec:est_DML_ROLS}
The Double/Debiased Machine Learning (DML) framework \citep{Chernozhukov_et_al_2018_DebiasedML} introduces the concept of Neyman-orthogonality (NO) in influence functions, enabling the estimation of the averate treatment effect (ATE) under non-linear but additively separable confounding. In this section, we demonstrate that estimates based on the NO moment condition for the partially linear model can be reinterpreted as the average partial effect (APE), provided the conditions of Theorem \ref{thm:main} are satisfied. This holds even though the partially linear model is misspecification under Assumption \ref{assump_DGP}. Crucially, the DML framework facilitates the estimation of the APE without requiring explicit knowledge of the treatment data generating process, allowing for the use of machine learning methods to approximate $r(Z_i)=E[X_i \mid Z_i]$. Additionally, it offers a rigorous foundation for conducting inference on the APE by leveraging the theoretical results of \cite{Chernozhukov_et_al_2018_DebiasedML}. Compared to the two-step procedure described in Section \ref{sec:est_naive_ML_ROLS}, the DML-based approach yields more robust estimation and inference procedures for the APE.

The particular DML-approach we use is based on the partially linear model of \cite{Robinson_1988}:
\begin{align*}
    Y_i = \theta X_i + g(Z_i) + \varepsilon_i
\end{align*}
In this model, $\theta$ is of interest to the econometrician and can be interpreted as the ATE. We will demonstrate that, provided the conditions of Theorem \ref{thm:main} are satisfied, $\hat{\theta}$ estimated via the NO moment for the partially linear model can be interpreted as the APE, even though the partially linear model is misspecified under our assumptions.

The partialling-out based NO moment for the partially linear model is given by:
\begin{equation}
    \label{eqn:DML_partial_lin_moment}
    \psi(Z_i, \theta, \eta_{i}) = \left(Y_i - l(Z_i) - \theta (X_i-r(Z_i))\right)(X_i-r(Z_i)).
\end{equation}
where $l(Z_i)=E[Y_i \mid Z_i]$ and $\eta_{i}=\{l(Z_i), r(Z_i)\}$ are the nuisance functions that need to estimated from the data, in this case with machine learning. This moment is first-order invariant to deviations in nuisance functions by satisfying the following conditions: 
\begin{align*}
E[\partial_{l} \psi(Z_i, \theta, \eta_{i})] &= 0 \\
E[\partial_{r} \psi(Z_i, \theta, \eta_{i})] &= 0.
\end{align*}
The solution to the NO sample moment is given by
$$
    \hat{\theta} = \frac{E_n[(X_i - \hat{r}(Z_i)) (Y_i - \hat{l}(Z_i))]}{E_n[(X_i - \hat{r}(Z_i))^2]} 
$$
where $E_n[X_i]= \frac{1}{n} \sum_{i=1}^n X_i$ and $\hat{\eta}_{i}=\{\hat{l}(Z_i), \hat{r}(Z_i)\}$ are the estimated nuisance functions.

Comparing this estimator to the R-OLS estimator, we observe that the estimator residualises the treatment variable similarly to R-OLS. However, unlike R-OLS, it also residualises the outcome variable. Under the required assumptions, Theorem \ref{thm:main} demonstrates that the average partial effect (APE) is identified solely due to the residualization of the treatment variable. This suggests that the NO moment condition identifies the APE while accommodating slower convergence rates of the nuisance function estimators due to double residualization. This intuition is formalized in Lemmas \ref{lemma:DML_convergence} and \ref{lemma:DML_APE}, showing that $\hat{\theta}$ consistently estimates the APE under the conditions of Theorem \ref{thm:main}, cross-fitting, and sufficiently fast convergence of the nuisance function estimators.

\begin{lemma}[Convergence of partially linear DML]
    \label{lemma:DML_convergence}
    Let $X_i = r(Z_i) + \nu_i$ with $E[\nu_i l(Z_i)] = 0$, $E[|\nu_i^2|] < \infty$ and $E[|\nu_iY_i|] < \infty$. Assume that the estimators of nuisance parameters $\eta_{i} = \{l(Z_i), r(Z_i)\}$ are fit using cross-fitting, and their convergence rates satisfy $\sqrt{n}n^{-(\varphi_r + \varphi_l)} \to 0$ and $n^{-2\varphi_r} \to 0$. Then, the method-of-moments estimator based on moment condition (\ref{eqn:DML_partial_lin_moment}) satisfies:
    \begin{equation*}
        \hat{\theta} = \frac{E_n[(X_i - \hat{r}(Z_i)) (Y_i - \hat{l}(Z_i))]}{E_n[(X_i - \hat{r}(Z_i))^2]} \xrightarrow{p} \frac{E[\nu_i Y_i]}{E[\nu_i^2]}
    \end{equation*}
\end{lemma}

\begin{lemma}[DML estimation and inference for the APE]
    \label{lemma:DML_APE}
    Under the assumptions of Lemma \ref{lemma:DML_convergence} and Theorem \ref{thm:main}, the method-of-moments estimator based on equation (\ref{eqn:DML_partial_lin_moment}) satisfies:
    $$
        \sqrt{n} \left( \hat{\theta} - E\left[ \partial_{X_i} Y_i \right] \right) \xrightarrow{p} 0.
    $$
    where $E\left[ \partial_{X_i} Y_i\right]$ is the average partial effect of $X_i$ on $Y_i$. Under the additonal regularity conditions of Theorem 4.1 in \cite{Chernozhukov_et_al_2018_DebiasedML}, the estimates are asymptotically normal.
\end{lemma}

We conclude that $\hat{\theta}$, based on the Neyman-orthogonal moment condition for the partially linear model, provides a consistent and asymptotically normal estimate of the average partial effect, under the conditions of Lemma \ref{lemma:DML_APE}.

Applying the assumptions of Lemma \ref{lemma:DML_APE} to a cross-fit R-OLS estimator gives
\begin{align*}
    \sqrt{n}(\hat{\theta} - E[\partial_{X_i} Y_i]) 
    &\xrightarrow{p} \sqrt{n} \left( \frac{E[\nu_i Y_i]}{E[\nu_i^2]} + \frac{E[(r(Z_i) - \hat{r}(Z_i)) Y_i]}{E[\nu_i^2]} - E[\partial_{X_i} Y_i] \right) \\
    &= \sqrt{n} \frac{E[(r(Z_i) - \hat{r}(Z_i)) Y_i]}{E[\nu_i^2]} \numberthis \label{eq:bias_ROLS}
\end{align*}
which goes to $0$ if $\sqrt{n}n^{-\varphi_r} \xrightarrow{p} 0$, requiring a convergence rate of $\varphi_r > 1/2$ compared to the joint convergence rate of $\varphi_r + \varphi_l > 1/2$ in Lemma \ref{lemma:DML_APE}. Alternatively, the estimation error $r(Z_i) - \hat{r}(Z_i)$ needs to be unrelated to $Y_i$, for example via conditional mean independence. Unfortunately, cross-fitting can't guarantee this, since $r(Z_i) - \hat{r}(Z_i)$ can have "remnants" correlated with $Z_i$, creating a correlation with $Y_i$ and inducing bias.

The advantage of using a NO moment to estimate the APE is in its robustness to estimation errors in $\hat{r}(X_i)$ compared to ROLS. Let 
\begin{align*}
    Y_i &= 2(X_i + X_i^2) + Z_i^3 + \varepsilon_i \\
    X_i &= exp(Z_i) + \nu_i
\end{align*}
with $\varepsilon_i, \nu_i \sim N(0,1)$ and $Z_i \sim Unif(0,2)$. Figure \ref{fig:ROLS_vs_DML} shows the robustness of the NO moment condition when estimating the APE. The figure compares the estimates of the APE of R-OLS and DML over 200 simulations in which the nuisance functions are estimated with a neural network. The simulation design aims to reflect realistic challenges in nuisance function estimation, such as imperfect model training and limited sample size, to show the different robustness of R-OLS and DML. Variation between simulations is not only induced by different realisations of the errors $\nu_i$ and $\varepsilon_i$, but also by randomly chosing the number of training epochs for the neural networks predicting $l(Z_i)=E[Y_i \mid Z_i]$ and $r(Z_i)=E[X_i \mid Z_i]$. This creates artificial variation in the quality of the predictions of $l(Z_i)$ and $r(Z_i)$, showcasing the increased robustness of the NO moment condition compared to R-OLS. The quality of the predictions is measured by the correlation between $\hat{\nu_i}$ and $Z_i$, providing a simple measure of how well the neural network estimates $r(Z_i)$ and how much potential for bias exists. A measure of the quality with which $\hat{l}(Z_i)$ predicts $l(Z_i)$ is not shown, since the NO moment condition is utilising this prediction for debiasing and not for identification.

Figure \ref{fig:ROLS_vs_DML} shows the NO moment condition is robust to correlation between $\hat{\nu}_i$ and $Z_i$. The APE is estimated consistently even for imperfect estimation of $r(Z_i)$. The R-OLS estimator on the other hand is biased when the correlation between $\hat{\nu}_i$ and $Z_i$ is high, which is exactly the bias shown in equation (\ref{eq:bias_ROLS}). 

This confirms the theoretical results above and shows that the estimates given by the NO moment condition for the partially linear model (\ref{eqn:DML_partial_lin_moment}) can be interpreted as the APE, under the Assumptions of Lemma \ref{lemma:DML_APE}. 

\begin{figure}[h]
    \centering
    \caption{Comparison of APE Estimates of R-OLS and DML}
    \includegraphics[width=\textwidth]{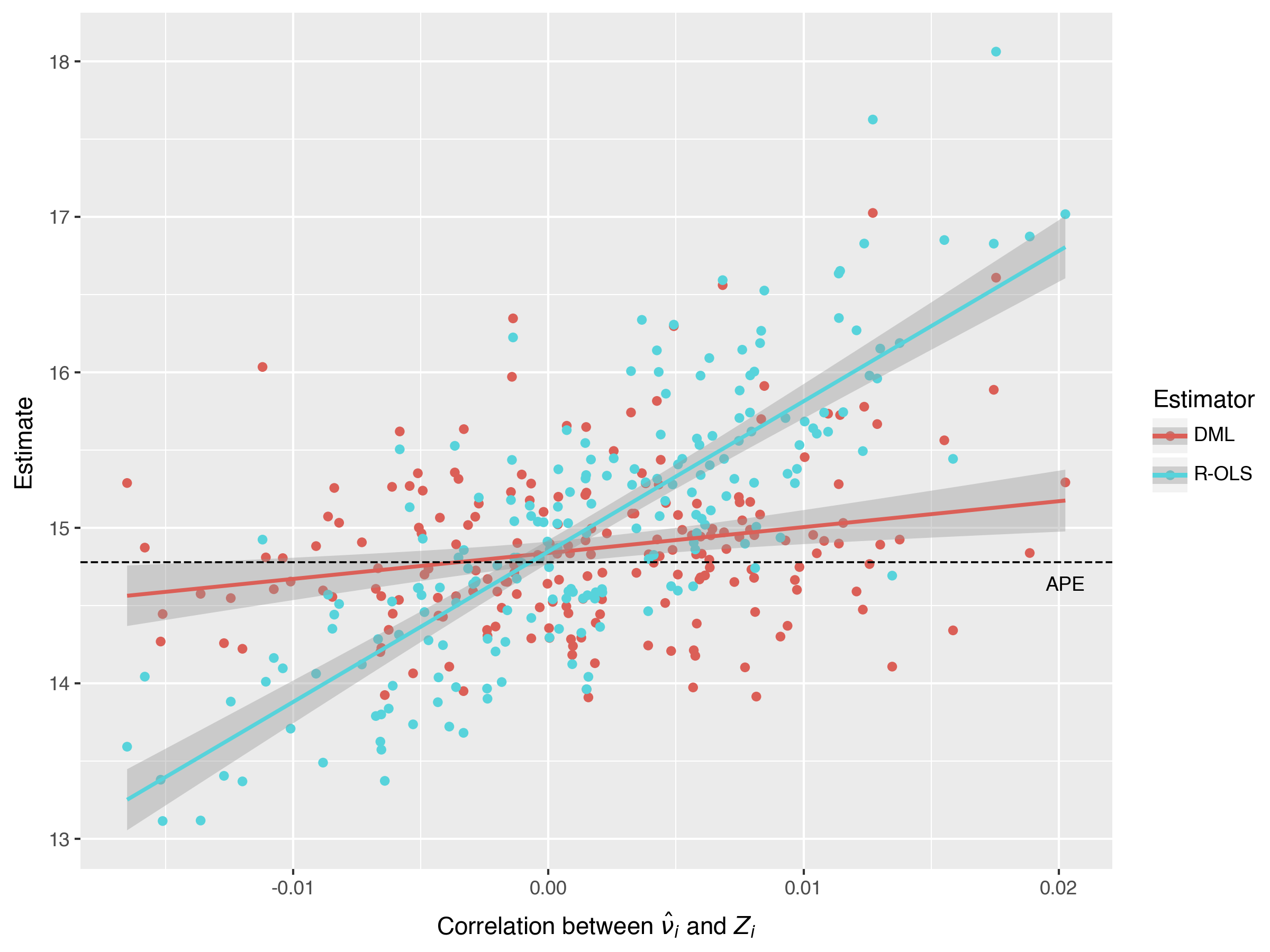}
    \label{fig:ROLS_vs_DML}
    \tablenote{Estimates of the APE for R-OLS and DML are shown. Number of simulations is 200, each with a sample size of 1000 and the number of cross-fitting splits set to 4. Both $l(Z_i)=E[Y_i \mid Z_i]$ and $r(Z_i)=E[X_i \mid Z_i]$ are predicted via neural networks implemented in the Scikit-Learn python package. The neural networks have three hidden layers, each with 64 nodes, and the number of training epochs in each simulation is drawn from a uniform distribution between 50 and 200. The x-axis shows the correlation between $\hat{\nu_i}$ and $Z_i$. The true APE is given by the black dashed line. The regressions lines show the linear fit of the estimated APE for each model and the correlation between $\hat{\nu_i}$ and $Z_i$, with $95\%$ confidence intervals shaded in grey. 
    }
\end{figure}

\section{Weights and Diagonistics} \label{sec:weights}
The OLS and IV estimand can be rewritten as weighted expectations of the partial derivative \citep{Yitzhaki1996, Angrist_1998, ANGRIST_KRUEGER_1999, Graham_Pinto_2022, kolesar2024dynamiccausaleffectsnonlinearworld}. Using a novel decomposition, Theorem \ref{thm:main} implicit uses this fact and make assumptions such that the weights ensure the APE is estimated. The following Lemma shows this for the R-OLS estimand, but the result applies to OLS immediately.\footnote{The weights applied to the OLS and IV estimand can be see in Appendix D.}

\begin{lemma}[R-OLS as weighted sum of partial derivatives]
    \label{lemma:weights_ROLS}
    Under Assumption \ref{assump_DGP}, the R-OLS estimand
    $$
        \beta = \frac{Cov(\nu_i, Y_i)}{Var(\nu_i)} 
    $$
    can be decomposed as a weighted sum of partial derivatives
    $$
        \beta = \sum_{m=1}^M \sum_{p=0}^{m-1} {m-1 \choose p} m E\left[ r\left(Z_{i}\right)^{m-1-p} \nu_{i}^{p} g_m\left(Z_{i}\right) \right] \frac{E\left[\nu_{i}^{p+2}\right] - E[\nu_{i}] E\left[\nu_{i}^{p+1}\right]}{ (p+1) Var(\nu_i) E\left[\nu_{i}^{p}\right]}
    $$
\end{lemma}

These weights are novel and show that one can interpret R-OLS estimates as a sum of partial derivatives of each element of the polynomial structure in the Y-DGP
$$
    m E\left[ r\left(Z_{i}\right)^{m-1-p} \nu_{i}^{p} g_m\left(Z_{i}\right) \right]
$$
with weights
$$
    \frac{E\left[\nu_{i}^{p+2}\right] - E[\nu_{i}] E\left[\nu_{i}^{p+1}\right]}{ (p+1) Var(\nu_i) E\left[\nu_{i}^{p}\right]}
$$
depending on the distribution of the treatment. 

In contrast to existing literature, which typically emphasizes weights applied to the partial derivative evaluated at specific values of the regressor, this decomposition enables a trade-off between the complexity of the outcome DGP and the distributional assumptions on the treatment variable. This trade-off is particularly evident in the R-OLS framework.

Applying the assumptions of R-OLS to these weights, we can see that Assumption \ref{assump_errors} guarantees that the R-OLS regression of $Y_i$ on $\nu_i$ weights the average partial effect of each polynomial in equation (\ref{DGP:Y}) by one. 

Building on this intuition, practitioners can identify whether the R-OLS estimator under- or overestimates the APE and determine the extent to which the underlying distributional assumptions hold, by estimating the sample analog of the above weights and examining their behavior across values of p. Specifically, the order of p up to which the weights approximate unity provide guidance on the appropriateness of R-OLS for estimating the APE under a given data-generating processes.

\section{Simulations} \label{sec:Simulations}
There are three variables in the simulation. A treatment variable $X$ for which the APE is estimated and two confounding variables, namely $Z = \{Z_1, Z_2\}$. The confounders $Z_1$ and $Z_2$ are observed and always $N(1,1)$ distributed. 

The assumption of normally distributed confounders is made for two reasons. Firstly, the complexity in the simulations does not originate from the distribution of the confounders, but from the DGPs used to specify $Y$ and $X$. Secondly, assuming normally distributed confounders facilitates the generation of data with higher-order interactions, without the concern of infinite variance.

One of the aspects highlighted in the simulations is the importance of the moments of the errors in the X-DPG, $\nu$, since Theorem \ref{thm:main} shows that they are crucial to estimate the APE with R-OLS. Hence, the simulations are split into two parts. In section \ref{sec:normal_errors} the errors satisfy Assumption \ref{assump_errors} by following $N(0,1)$. Then in section \ref{sec:non_normal_errors} this assumption is violated and $\nu$ is drawn from a Gaussian mixture satisfying only a limited number of moment conditions. In contrast to the $X$-DGP errors, the errors in the $Y$-DGP, denoted as $\varepsilon$, only need to satisfy an unconfoundedness condition and are therefore always generated from a standard normal distribution.

\begin{table}[h]
    \centering
    \caption{Data generating processes for the simulations}
    \begin{tabularx}{\textwidth}{@{}l|Y|Y@{}}
                            & \bf{Y} & \bm{$X$} \\ \hline 
        \bf{additive}       & $X + Z_1 + Z_2 + \varepsilon$ & $Z_1 + Z_2 + \nu$ \\ \hline
        \bf{simple}         & $\sum_{m=0}^{M} X^m Z_1 Z_2 + \varepsilon $ & $Z_1 Z_2 + \nu$ \\ \hline
        \bf{complex}  & $\sum_{m=0}^{M} X^m cos(Z_1) sin(Z_2) + \varepsilon$ & $5sin(Z_1)cos(Z_2) + \nu$ \\ 
    \end{tabularx}
    \tablenote{$Z_1$ and $Z_2$ are $N(1,1)$ distributed. $\varepsilon$ is $N(0,1)$ distributed. $\nu$ is either $N(0,1)$ distributed or follows a Gaussian mixture with distribution $0.5 \cdot N(\mu, 1-\mu^2) + 0.5 \cdot N(-\mu, 1-\mu^2)$ and $\mu = 0.9$. M is the order of polynomials and changes between simulations.} 
    \label{table:DGPs}
\end{table}
\begin{table}
    \centering
    \caption{Models estimated in the simulations}
    \begin{tabularx}{\textwidth}{@{}l|Y|Y@{}}
       & \bf{Estimated Model} & \bf{Estimated APE}  \\ \hline
      \bf{simple OLS} & $Y \sim \beta_0 + \beta_1X + \beta_2Z_1 + \beta_3Z_2$ & $\hat{\beta}_1$ \\ \hline
      \bf{interacted OLS} & $Y \sim \beta_0 + Poly(\{X, Z_1, Z_2\}, \text{Degree}=3)$ & $\frac{1}{N} \sum_{i=1}^N \frac{\partial \widehat{Poly}(\{X_i, Z_{i1}, Z_{i2}\}, \text{Degree}=3)}{\partial X_i}$ \\ \hline
      \bf{PL-GAM} & $Y \sim \beta_0 + \beta_1X1 + BSpline_3(Z_1) + BSpline_3(Z_2)$ & $\hat{\beta}_1$ \\ \hline
      \bf{R-OLS} & \makecell{ 1) $\hat{\nu}=X-\hat{X}$ with NN \\ 2)$Y \sim \beta_0 + \beta_1\hat{\nu}$ } & $\hat{\beta}_1$
    \end{tabularx}
    \tablenote{$Poly(X, Degree)$ symbolises that all interactions of variables in $X$ are created up to order $Degree$. PL-GAM is a semi-parametric model in which the non-parametric part is estimated with generalised additive models. Third-degree BSplines are used for the non-parametric but additive part. NN in R-OLS symbolises that a neural network is used to residualise $X$ by predicting $\hat{X}$ with cross-fitting. The NN has three layers, each with 64 nodes and 'relu' activation function, and was trained for 500 epochs with the 'adam' optimiser in the "tensorflow" Python library.} 
    \label{table:Models}
\end{table}

Table \ref{table:DGPs} shows the DGPs used to generate the data in each specification of the simulations. For some of the DGPs, sine and cosine are used to achieve highly non-linear functional forms while keeping the variables bounded and non-normally distributed, apart from the error component. 

All the estimated models compared in the simulations can be seen in Table \ref{table:Models}. PL-GAM is a semi-parametric model in which the non-parametric part is estimated with generalised additive models using third-degree BSplines. The R-OLS estimator predicts the residuals by first training a neural network with three hidden layers, each with 64 nodes, to predict $X$ for each individual. Then $\nu$ is approximated via $\hat{\nu}=X-\hat{X}$ on the same sample the algorithm was trained on.\footnote{I used the "tensorflow" library in Python to train the neural networks on an Nvidia Tesla 4 GPU. During training, no cross-validated grid search for the optimal hyperparameters of the neural network is done due to the high computational cost associated. Hence, the simulations show the lower bound of the potential performance of R-OLS.} 

Throughout the next sections, the results of some simulations are highlighted. The appendix contains all remaining specifications.
 
\subsection{Standard Normal Errors}\label{sec:normal_errors}
\subsubsection{Simple OLS Assumptions Satisfied}
We begin with simulations showing the results for the additive $Y$-DGP and a variation of $X$-DGPs with standard normal errors. In these simulations, the assumptions of simple OLS are satisfied by the $Y$-DGP and hence any OLS estimator is unbiased. 

The results can be seen in Table \ref{table:additive_Y}. The table shows 10000 simulations for all five estimated models under different sizes of the dataset, $N$, and $X$-DGPs. Each field reports the average of the estimated APEs, specified in Table \ref{table:Models}, with its standard deviation in round brackets and mean squared error in square brackets. 

The results show that all estimators perform well in this case due to the simplicity of the simulation. The simple OLS, PL-GAM and interacted OLS estimators are unbiased and have virtually no variance. The R-OLS estimator is unbiased, but has higher variance. 

These simulations highlight that if the assumptions of simple OLS are satisfied, R-OLS obviously does not improve on the simple OLS estimator. However, it is important to note that the assumptions of simple OLS are rarely satisfied in practice and R-OLS offers increased robustness at the price of increased variance. The next simulations show the results for more complex $Y$-DGPs in which the assumptions of simple OLS are not satisfied anymore, leading to biased estimates.

\begin{sidewaystable}
    \caption{APE estimation with additive $Y$-DGP and a variation of $X$-DGPs with standard normal errors}
    \resizebox{\textwidth}{!}{
    \begin{tabular}{llcccclcccclcccc}
                                    &  & \multicolumn{4}{c}{\textbf{additive $X$-DGP with APE=1}}                                                                                                                                                                                                                                                                                                       &  & \multicolumn{4}{c}{\textbf{simple $X$-DGP with APE=1}}                                                                                                                                                                                                                                                                                                      &  & \multicolumn{4}{c}{\textbf{complex $X$-DGP with APE=1}}                                                                                                                                                                                                                                                                                                            \\ \cline{3-6} \cline{8-11} \cline{13-16} 
                                    &  & \multicolumn{1}{c|}{N=100}                                                               & \multicolumn{1}{c|}{N=500}                                                               & \multicolumn{1}{c|}{N=1000}                                                              & N=5000                                                              &  & \multicolumn{1}{c|}{N=100}                                                               & \multicolumn{1}{c|}{N=500}                                                               & \multicolumn{1}{c|}{N=1000}                                                              & N=5000                                                              &  & \multicolumn{1}{c|}{N=100}                                                                  & \multicolumn{1}{c|}{N=500}                                                                & \multicolumn{1}{c|}{N=1000}                                                               & N=5000                                                               \\ \cline{3-6} \cline{8-11} \cline{13-16} 
    \textbf{simple OLS}          &  & \multicolumn{1}{c|}{\begin{tabular}[c]{@{}c@{}}1.0\\ (0.0)\\ {[}0.0{]}\end{tabular}} & \multicolumn{1}{c|}{\begin{tabular}[c]{@{}c@{}}1.0\\ (0.0)\\ {[}0.0{]}\end{tabular}}  & \multicolumn{1}{c|}{\begin{tabular}[c]{@{}c@{}}1.0\\ (0.0)\\ {[}0.0{]}\end{tabular}}           & \begin{tabular}[c]{@{}c@{}}1.0\\ (0.0)\\ {[}0.0{]}\end{tabular}     &  & \multicolumn{1}{c|}{\begin{tabular}[c]{@{}c@{}}1.0\\ (0.0)\\ {[}0.0{]}\end{tabular}}    & \multicolumn{1}{c|}{\begin{tabular}[c]{@{}c@{}}1.0\\ (0.0)\\ {[}0.0{]}\end{tabular}} & \multicolumn{1}{c|}{\begin{tabular}[c]{@{}c@{}}1.0\\ (0.0)\\ {[}0.0{]}\end{tabular}} & \begin{tabular}[c]{@{}c@{}}1.0\\ (0.0)\\ {[}0.0{]}\end{tabular}          &  & \multicolumn{1}{c|}{\begin{tabular}[c]{@{}c@{}}1.0\\ (0.0)\\ {[}0.0{]}\end{tabular}}  & \multicolumn{1}{c|}{\begin{tabular}[c]{@{}c@{}}1.0\\ (0.0)\\ {[}0.0{]}\end{tabular}}   & \multicolumn{1}{c|}{\begin{tabular}[c]{@{}c@{}}1.0\\ (0.0)\\ {[}0.0{]}\end{tabular}}   & \begin{tabular}[c]{@{}c@{}}1.0\\ (0.0)\\ {[}0.0{]}\end{tabular} \\ \cline{1-1} \cline{3-6} \cline{8-11} \cline{13-16} 
    \textbf{interacted OLS}      &  & \multicolumn{1}{c|}{\begin{tabular}[c]{@{}c@{}}1.0\\ (0.0)\\ {[}0.0{]}\end{tabular}} & \multicolumn{1}{c|}{\begin{tabular}[c]{@{}c@{}}1.0\\ (0.0)\\ {[}0.0{]}\end{tabular}}  & \multicolumn{1}{c|}{\begin{tabular}[c]{@{}c@{}}1.0\\ (0.0)\\ {[}0.0{]}\end{tabular}}           & \begin{tabular}[c]{@{}c@{}}1.0\\ (0.0)\\ {[}0.0{]}\end{tabular}     &  & \multicolumn{1}{c|}{\begin{tabular}[c]{@{}c@{}}1.0\\ (0.0)\\ {[}0.0{]}\end{tabular}}    & \multicolumn{1}{c|}{\begin{tabular}[c]{@{}c@{}}1.0\\ (0.0)\\ {[}0.0{]}\end{tabular}} & \multicolumn{1}{c|}{\begin{tabular}[c]{@{}c@{}}1.0\\ (0.0)\\ {[}0.0{]}\end{tabular}} & \begin{tabular}[c]{@{}c@{}}1.0\\ (0.0)\\ {[}0.0{]}\end{tabular}          &  & \multicolumn{1}{c|}{\begin{tabular}[c]{@{}c@{}}1.0\\ (0.0)\\ {[}0.0{]}\end{tabular}} & \multicolumn{1}{c|}{\begin{tabular}[c]{@{}c@{}}1.0\\ (0.0)\\ {[}0.0{]}\end{tabular}}   & \multicolumn{1}{c|}{\begin{tabular}[c]{@{}c@{}}1.0\\ (0.0)\\ {[}0.0{]}\end{tabular}}   & \begin{tabular}[c]{@{}c@{}}1.0\\ (0.0)\\ {[}0.0{]}\end{tabular}  \\ \cline{1-1} \cline{3-6} \cline{8-11} \cline{13-16} 
    \textbf{PL-GAM}              &  & \multicolumn{1}{c|}{\begin{tabular}[c]{@{}c@{}}1.0\\ (0.0)\\ {[}0.0{]}\end{tabular}} & \multicolumn{1}{c|}{\begin{tabular}[c]{@{}c@{}}1.0\\ (0.0)\\ {[}0.0{]}\end{tabular}}  & \multicolumn{1}{c|}{\begin{tabular}[c]{@{}c@{}}1.0\\ (0.0)\\ {[}0.0{]}\end{tabular}}           & \begin{tabular}[c]{@{}c@{}}1.0\\ (0.0)\\ {[}0.0{]}\end{tabular}     &  & \multicolumn{1}{c|}{\begin{tabular}[c]{@{}c@{}}1.0\\ (0.0)\\ {[}0.0{]}\end{tabular}}    & \multicolumn{1}{c|}{\begin{tabular}[c]{@{}c@{}}1.0\\ (0.0)\\ {[}0.0{]}\end{tabular}} & \multicolumn{1}{c|}{\begin{tabular}[c]{@{}c@{}}1.0\\ (0.0)\\ {[}0.0{]}\end{tabular}} & \begin{tabular}[c]{@{}c@{}}1.0\\ (0.0)\\ {[}0.0{]}\end{tabular}          &  & \multicolumn{1}{c|}{\begin{tabular}[c]{@{}c@{}}1.0\\ (0.0)\\ {[}0.0{]}\end{tabular}}  & \multicolumn{1}{c|}{\begin{tabular}[c]{@{}c@{}}1.0\\ (0.0)\\ {[}0.0{]}\end{tabular}}   & \multicolumn{1}{c|}{\begin{tabular}[c]{@{}c@{}}1.0\\ (0.0)\\ {[}0.0{]}\end{tabular}}   & \begin{tabular}[c]{@{}c@{}}1.0\\ (0.0)\\ {[}0.0{]}\end{tabular}  \\ \cline{1-1} \cline{3-6} \cline{8-11} \cline{13-16} 
    \textbf{R-OLS}               &  & \multicolumn{1}{c|}{\begin{tabular}[c]{@{}c@{}}1.05\\ (0.14)\\ {[}0.02{]}\end{tabular}} & \multicolumn{1}{c|}{\begin{tabular}[c]{@{}c@{}}1.0\\ (0.0)\\ {[}0.0{]}\end{tabular}}  & \multicolumn{1}{c|}{\begin{tabular}[c]{@{}c@{}}1.0\\ (0.0)\\ {[}0.0{]}\end{tabular}}        & \begin{tabular}[c]{@{}c@{}}1.0\\ (0.0)\\ {[}0.0{]}\end{tabular}     &  & \multicolumn{1}{c|}{\begin{tabular}[c]{@{}c@{}}1.02\\ (0.05)\\ {[}0.0{]}\end{tabular}}    & \multicolumn{1}{c|}{\begin{tabular}[c]{@{}c@{}}1.0\\ (0.01)\\ {[}0.0{]}\end{tabular}} & \multicolumn{1}{c|}{\begin{tabular}[c]{@{}c@{}}1.0\\ (0.01)\\ {[}0.0{]}\end{tabular}} & \begin{tabular}[c]{@{}c@{}}1.0\\ (0.0)\\ {[}0.0{]}\end{tabular}      &  & \multicolumn{1}{c|}{\begin{tabular}[c]{@{}c@{}}1.0\\ (0.01)\\ {[}0.0{]}\end{tabular}}  & \multicolumn{1}{c|}{\begin{tabular}[c]{@{}c@{}}1.0\\ (0.01)\\ {[}0.0{]}\end{tabular}}   & \multicolumn{1}{c|}{\begin{tabular}[c]{@{}c@{}}1.0\\ (0.01)\\ {[}0.0{]}\end{tabular}}   & \begin{tabular}[c]{@{}c@{}}1.0\\ (0.0)\\ {[}0.0{]}\end{tabular} 
    \end{tabular}
    }
    \tablenote{Mean APE estimate after 10000 simulations. Round brackets show the standard deviation and square brackets the MSE of the estimates. $N$ is the sample size. Within each simulation, the same data is used for all estimators. Details of the DGP and estimated models are shown in Table \ref{table:DGPs} and \ref{table:Models}, respectively. Errors in the $X$-DGP follow a standard normal distribution.}
    \label{table:additive_Y}
\end{sidewaystable}

\subsubsection{Simple OLS Assumptions Violated}
Results of the simulation with the simple $X$-DGP and the complex $Y$-DGP and standard normal errors can be seen in Table \ref{table:complex_Y_simple_X}. 
Simple OLS and PL-GAM severely underestimate the APE in this case and for some specifications even estimate the opposite sign. This comes from a negative correlation between the residual and the APE. Interacted OLS consistently estimates the APE since from a Frisch-Waugh-Lovell perspective the estimated equation residualises $X$ correctly, essentially replicating R-OLS, just with a linear regression.\footnote{Table \ref{table:complex_Y_additive_X} in the appendix compares simulations in which all estimators replicate R-OLS from a FWL perspective and shows unbiasedness and similar variance for all estimators.} R-OLS, in which a neural network estimates the errors, does well in larger samples and/or low $M$. 

\begin{sidewaystable}
    \caption{APE estimation with complex $Y$ and simple $X$-DGPs and standard normal errors}
    \resizebox{\textwidth}{!}{
    \begin{tabular}{llcccclcccclcccc}
                                    &  & \multicolumn{4}{c}{\textbf{M=1 with APE=0.17}}                                                                                                                                                                                                                                                                                                       &  & \multicolumn{4}{c}{\textbf{M=2 with APE=-0.14}}                                                                                                                                                                                                                                                                                                      &  & \multicolumn{4}{c}{\textbf{M=3 with APE=-2.06}}                                                                                                                                                                                                                                                                                                            \\ \cline{3-6} \cline{8-11} \cline{13-16} 
                                    &  & \multicolumn{1}{c|}{N=100}                                                               & \multicolumn{1}{c|}{N=500}                                                               & \multicolumn{1}{c|}{N=1000}                                                              & N=5000                                                              &  & \multicolumn{1}{c|}{N=100}                                                               & \multicolumn{1}{c|}{N=500}                                                               & \multicolumn{1}{c|}{N=1000}                                                              & N=5000                                                              &  & \multicolumn{1}{c|}{N=100}                                                                  & \multicolumn{1}{c|}{N=500}                                                                & \multicolumn{1}{c|}{N=1000}                                                               & N=5000                                                               \\ \cline{3-6} \cline{8-11} \cline{13-16} 
    \textbf{simple OLS}          &  & \multicolumn{1}{c|}{\begin{tabular}[c]{@{}c@{}}-0.08\\ (0.14)\\ {[}0.08{]}\end{tabular}} & \multicolumn{1}{c|}{\begin{tabular}[c]{@{}c@{}}-0.08\\ (0.06)\\ {[}0.07{]}\end{tabular}} & \multicolumn{1}{c|}{\begin{tabular}[c]{@{}c@{}}-0.08\\ (0.04)\\ {[}0.06{]}\end{tabular}} & \begin{tabular}[c]{@{}c@{}}-0.08\\ (0.02)\\ {[}0.06{]}\end{tabular} &  & \multicolumn{1}{c|}{\begin{tabular}[c]{@{}c@{}}-1.05\\ (0.96)\\ {[}1.76{]}\end{tabular}} & \multicolumn{1}{c|}{\begin{tabular}[c]{@{}c@{}}-1.07\\ (0.46)\\ {[}1.08{]}\end{tabular}} & \multicolumn{1}{c|}{\begin{tabular}[c]{@{}c@{}}-1.07\\ (0.34)\\ {[}0.98{]}\end{tabular}} & \begin{tabular}[c]{@{}c@{}}-1.07\\ (0.15)\\ {[}0.89{]}\end{tabular} &  & \multicolumn{1}{c|}{\begin{tabular}[c]{@{}c@{}}-5.08\\ (10.08)\\ {[}110.72{]}\end{tabular}} & \multicolumn{1}{c|}{\begin{tabular}[c]{@{}c@{}}-5.01\\ (4.99)\\ {[}33.64{]}\end{tabular}} & \multicolumn{1}{c|}{\begin{tabular}[c]{@{}c@{}}-4.94\\ (3.8)\\ {[}22.76{]}\end{tabular}}  & \begin{tabular}[c]{@{}c@{}}-4.98\\ (1.67)\\ {[}11.31{]}\end{tabular} \\ \cline{1-1} \cline{3-6} \cline{8-11} \cline{13-16} 
    \textbf{interacted OLS}      &  & \multicolumn{1}{c|}{\begin{tabular}[c]{@{}c@{}}0.22\\ (0.1)\\ {[}0.01{]}\end{tabular}}   & \multicolumn{1}{c|}{\begin{tabular}[c]{@{}c@{}}0.18\\ (0.04)\\ {[}0.0{]}\end{tabular}}   & \multicolumn{1}{c|}{\begin{tabular}[c]{@{}c@{}}0.18\\ (0.03)\\ {[}0.0{]}\end{tabular}}   & \begin{tabular}[c]{@{}c@{}}0.17\\ (0.01)\\ {[}0.0{]}\end{tabular}   &  & \multicolumn{1}{c|}{\begin{tabular}[c]{@{}c@{}}-0.01\\ (0.48)\\ {[}0.25{]}\end{tabular}} & \multicolumn{1}{c|}{\begin{tabular}[c]{@{}c@{}}-0.09\\ (0.22)\\ {[}0.05{]}\end{tabular}} & \multicolumn{1}{c|}{\begin{tabular}[c]{@{}c@{}}-0.11\\ (0.15)\\ {[}0.02{]}\end{tabular}} & \begin{tabular}[c]{@{}c@{}}-0.13\\ (0.07)\\ {[}0.01{]}\end{tabular} &  & \multicolumn{1}{c|}{\begin{tabular}[c]{@{}c@{}}-1.49\\ (3.58)\\ {[}13.15{]}\end{tabular}}   & \multicolumn{1}{c|}{\begin{tabular}[c]{@{}c@{}}-1.78\\ (1.69)\\ {[}2.93{]}\end{tabular}}  & \multicolumn{1}{c|}{\begin{tabular}[c]{@{}c@{}}-1.87\\ (1.22)\\ {[}1.53{]}\end{tabular}}  & \begin{tabular}[c]{@{}c@{}}-2.01\\ (0.58)\\ {[}0.34{]}\end{tabular}  \\ \cline{1-1} \cline{3-6} \cline{8-11} \cline{13-16} 
    \textbf{PL-GAM}              &  & \multicolumn{1}{c|}{\begin{tabular}[c]{@{}c@{}}-0.07\\ (0.13)\\ {[}0.07{]}\end{tabular}} & \multicolumn{1}{c|}{\begin{tabular}[c]{@{}c@{}}-0.09\\ (0.05)\\ {[}0.07{]}\end{tabular}} & \multicolumn{1}{c|}{\begin{tabular}[c]{@{}c@{}}-0.09\\ (0.04)\\ {[}0.07{]}\end{tabular}} & \begin{tabular}[c]{@{}c@{}}-0.08\\ (0.02)\\ {[}0.06{]}\end{tabular} &  & \multicolumn{1}{c|}{\begin{tabular}[c]{@{}c@{}}-1.01\\ (0.7)\\ {[}1.26{]}\end{tabular}}  & \multicolumn{1}{c|}{\begin{tabular}[c]{@{}c@{}}-1.13\\ (0.36)\\ {[}1.11{]}\end{tabular}} & \multicolumn{1}{c|}{\begin{tabular}[c]{@{}c@{}}-1.12\\ (0.27)\\ {[}1.04{]}\end{tabular}} & \begin{tabular}[c]{@{}c@{}}-1.09\\ (0.14)\\ {[}0.93{]}\end{tabular} &  & \multicolumn{1}{c|}{\begin{tabular}[c]{@{}c@{}}-5.62\\ (5.08)\\ {[}38.51{]}\end{tabular}}   & \multicolumn{1}{c|}{\begin{tabular}[c]{@{}c@{}}-5.82\\ (3.15)\\ {[}24.1{]}\end{tabular}}  & \multicolumn{1}{c|}{\begin{tabular}[c]{@{}c@{}}-5.57\\ (2.58)\\ {[}19.04{]}\end{tabular}} & \begin{tabular}[c]{@{}c@{}}-5.18\\ (1.4)\\ {[}11.69{]}\end{tabular}  \\ \cline{1-1} \cline{3-6} \cline{8-11} \cline{13-16} 
    \textbf{R-OLS}        &  & \multicolumn{1}{c|}{\begin{tabular}[c]{@{}c@{}}0.18\\ (0.13)\\ {[}0.02{]}\end{tabular}}  & \multicolumn{1}{c|}{\begin{tabular}[c]{@{}c@{}}0.17\\ (0.04)\\ {[}0.0{]}\end{tabular}}   & \multicolumn{1}{c|}{\begin{tabular}[c]{@{}c@{}}0.17\\ (0.03)\\ {[}0.0{]}\end{tabular}}   & \begin{tabular}[c]{@{}c@{}}0.17\\ (0.02)\\ {[}0.0{]}\end{tabular}   &  & \multicolumn{1}{c|}{\begin{tabular}[c]{@{}c@{}}-0.05\\ (0.69)\\ {[}0.49{]}\end{tabular}}  & \multicolumn{1}{c|}{\begin{tabular}[c]{@{}c@{}}-0.06\\ (0.26)\\ {[}0.07{]}\end{tabular}} & \multicolumn{1}{c|}{\begin{tabular}[c]{@{}c@{}}-0.12\\ (0.19)\\ {[}0.03{]}\end{tabular}} & \begin{tabular}[c]{@{}c@{}}-0.13\\ (0.08)\\ {[}0.01{]}\end{tabular} &  & \multicolumn{1}{c|}{\begin{tabular}[c]{@{}c@{}}-1.17\\ (7.41)\\ {[}55.66{]}\end{tabular}}    & \multicolumn{1}{c|}{\begin{tabular}[c]{@{}c@{}}-1.93\\ (2.13)\\ {[}4.55{]}\end{tabular}}  & \multicolumn{1}{c|}{\begin{tabular}[c]{@{}c@{}}-1.77\\ (1.89)\\ {[}3.67{]}\end{tabular}}  & \begin{tabular}[c]{@{}c@{}}-1.98\\ (0.74)\\ {[}0.55{]}\end{tabular} 
    \end{tabular}
    }
    \tablenote{Mean APE estimate after 10000 simulations. Round brackets show the standard deviation and square brackets the MSE of the estimates. $M$ is the order of polynomials in the $Y$-DGP and $N$ the sample size. Within each simulation, the same data is used for all estimators. Details of the DGP and estimated models are shown in Table \ref{table:DGPs} and \ref{table:Models}, respectively. Errors in the $X$-DGP follow a standard normal distribution.}
    \label{table:complex_Y_simple_X}
\end{sidewaystable}

Moving to a more complicated DGP for $X$ in Table \ref{table:complex_Y_complex_X}, we can see that neither simple OLS nor PL-GAM can provide a consistent estimate of the APE since they can neither fit the $X$ nor the $Y$-DGP well. Furthermore, due to the increased complexity of the $X$-DGP, interacted OLS is not able to replicate the idea behind R-OLS anymore and its estimates are off, but not by much, hinting that the interacted model is surprisingly good at approximating the DGPs. In this setting, R-OLS with neural networks works very well if $M=2$. Surprisingly, for $M=1$ and $M=3$, the estimation gets worse with an increased sample size. However, it still improves on non-R-OLS alternatives. 

\begin{sidewaystable}[]
    \caption{APE estimation with complex $Y$ and complex $X$-DGPs and standard normal errors}
    \resizebox{\textwidth}{!}{
    \begin{tabular}{llcccclcccclcccc}
                                 &  & \multicolumn{4}{c}{\textbf{M=1 with APE=0.17}}                                                                                                                                                                                                                                                                                                &  & \multicolumn{4}{c}{\textbf{M=2 with APE=0.21}}                                                                                                                                                                                                                                                                                                   &  & \multicolumn{4}{c}{\textbf{M=3 with APE=1.52}}                                                                                                                                                                                                                                                                                                   \\ \cline{3-6} \cline{8-11} \cline{13-16} 
                                 &  & \multicolumn{1}{c|}{N=100}                                                              & \multicolumn{1}{c|}{N=500}                                                             & \multicolumn{1}{c|}{N=1000}                                                            & N=5000                                                            &  & \multicolumn{1}{c|}{N=100}                                                              & \multicolumn{1}{c|}{N=500}                                                              & \multicolumn{1}{c|}{N=1000}                                                             & N=5000                                                             &  & \multicolumn{1}{c|}{N=100}                                                              & \multicolumn{1}{c|}{N=500}                                                              & \multicolumn{1}{c|}{N=1000}                                                             & N=5000                                                             \\ \cline{3-6} \cline{8-11} \cline{13-16} 
    \textbf{simple OLS}          &  & \multicolumn{1}{c|}{\begin{tabular}[c]{@{}c@{}}0.13\\ (0.05)\\ {[}0.0{]}\end{tabular}}  & \multicolumn{1}{c|}{\begin{tabular}[c]{@{}c@{}}0.13\\ (0.02)\\ {[}0.0{]}\end{tabular}} & \multicolumn{1}{c|}{\begin{tabular}[c]{@{}c@{}}0.13\\ (0.02)\\ {[}0.0{]}\end{tabular}} & \begin{tabular}[c]{@{}c@{}}0.13\\ (0.01)\\ {[}0.0{]}\end{tabular} &  & \multicolumn{1}{c|}{\begin{tabular}[c]{@{}c@{}}0.25\\ (0.17)\\ {[}0.03{]}\end{tabular}} & \multicolumn{1}{c|}{\begin{tabular}[c]{@{}c@{}}0.24\\ (0.07)\\ {[}0.01{]}\end{tabular}} & \multicolumn{1}{c|}{\begin{tabular}[c]{@{}c@{}}0.24\\ (0.05)\\ {[}0.0{]}\end{tabular}}  & \begin{tabular}[c]{@{}c@{}}0.24\\ (0.02)\\ {[}0.0{]}\end{tabular}  &  & \multicolumn{1}{c|}{\begin{tabular}[c]{@{}c@{}}1.06\\ (0.77)\\ {[}0.8{]}\end{tabular}}  & \multicolumn{1}{c|}{\begin{tabular}[c]{@{}c@{}}1.04\\ (0.33)\\ {[}0.34{]}\end{tabular}} & \multicolumn{1}{c|}{\begin{tabular}[c]{@{}c@{}}1.03\\ (0.24)\\ {[}0.29{]}\end{tabular}} & \begin{tabular}[c]{@{}c@{}}1.02\\ (0.1)\\ {[}0.26{]}\end{tabular}  \\ \cline{1-1} \cline{3-6} \cline{8-11} \cline{13-16} 
    \textbf{interacted OLS}      &  & \multicolumn{1}{c|}{\begin{tabular}[c]{@{}c@{}}0.17\\ (0.08)\\ {[}0.01{]}\end{tabular}} & \multicolumn{1}{c|}{\begin{tabular}[c]{@{}c@{}}0.17\\ (0.05)\\ {[}0.0{]}\end{tabular}} & \multicolumn{1}{c|}{\begin{tabular}[c]{@{}c@{}}0.18\\ (0.04)\\ {[}0.0{]}\end{tabular}} & \begin{tabular}[c]{@{}c@{}}0.18\\ (0.03)\\ {[}0.0{]}\end{tabular} &  & \multicolumn{1}{c|}{\begin{tabular}[c]{@{}c@{}}0.29\\ (0.33)\\ {[}0.12{]}\end{tabular}} & \multicolumn{1}{c|}{\begin{tabular}[c]{@{}c@{}}0.31\\ (0.17)\\ {[}0.04{]}\end{tabular}} & \multicolumn{1}{c|}{\begin{tabular}[c]{@{}c@{}}0.33\\ (0.13)\\ {[}0.03{]}\end{tabular}} & \begin{tabular}[c]{@{}c@{}}0.33\\ (0.07)\\ {[}0.02{]}\end{tabular} &  & \multicolumn{1}{c|}{\begin{tabular}[c]{@{}c@{}}1.63\\ (1.45)\\ {[}2.11{]}\end{tabular}} & \multicolumn{1}{c|}{\begin{tabular}[c]{@{}c@{}}1.55\\ (0.67)\\ {[}0.45{]}\end{tabular}} & \multicolumn{1}{c|}{\begin{tabular}[c]{@{}c@{}}1.58\\ (0.49)\\ {[}0.25{]}\end{tabular}} & \begin{tabular}[c]{@{}c@{}}1.63\\ (0.25)\\ {[}0.07{]}\end{tabular} \\ \cline{1-1} \cline{3-6} \cline{8-11} \cline{13-16} 
    \textbf{PL-GAM}              &  & \multicolumn{1}{c|}{\begin{tabular}[c]{@{}c@{}}0.14\\ (0.07)\\ {[}0.01{]}\end{tabular}} & \multicolumn{1}{c|}{\begin{tabular}[c]{@{}c@{}}0.12\\ (0.02)\\ {[}0.0{]}\end{tabular}} & \multicolumn{1}{c|}{\begin{tabular}[c]{@{}c@{}}0.11\\ (0.02)\\ {[}0.0{]}\end{tabular}} & \begin{tabular}[c]{@{}c@{}}0.11\\ (0.01)\\ {[}0.0{]}\end{tabular} &  & \multicolumn{1}{c|}{\begin{tabular}[c]{@{}c@{}}0.19\\ (0.21)\\ {[}0.04{]}\end{tabular}} & \multicolumn{1}{c|}{\begin{tabular}[c]{@{}c@{}}0.15\\ (0.08)\\ {[}0.01{]}\end{tabular}} & \multicolumn{1}{c|}{\begin{tabular}[c]{@{}c@{}}0.15\\ (0.05)\\ {[}0.01{]}\end{tabular}} & \begin{tabular}[c]{@{}c@{}}0.14\\ (0.02)\\ {[}0.0{]}\end{tabular}  &  & \multicolumn{1}{c|}{\begin{tabular}[c]{@{}c@{}}1.09\\ (1.04)\\ {[}1.27{]}\end{tabular}} & \multicolumn{1}{c|}{\begin{tabular}[c]{@{}c@{}}0.99\\ (0.39)\\ {[}0.44{]}\end{tabular}} & \multicolumn{1}{c|}{\begin{tabular}[c]{@{}c@{}}0.95\\ (0.27)\\ {[}0.4{]}\end{tabular}}  & \begin{tabular}[c]{@{}c@{}}0.92\\ (0.12)\\ {[}0.37{]}\end{tabular} \\ \cline{1-1} \cline{3-6} \cline{8-11} \cline{13-16} 
    \textbf{R-OLS}        &  & \multicolumn{1}{c|}{\begin{tabular}[c]{@{}c@{}}0.16\\ (0.11)\\ {[}0.03{]}\end{tabular}} & \multicolumn{1}{c|}{\begin{tabular}[c]{@{}c@{}}0.17\\ (0.04)\\ {[}0.0{]}\end{tabular}} & \multicolumn{1}{c|}{\begin{tabular}[c]{@{}c@{}}0.15\\ (0.03)\\ {[}0.0{]}\end{tabular}} & \begin{tabular}[c]{@{}c@{}}0.13\\ (0.01)\\ {[}0.0{]}\end{tabular} &  & \multicolumn{1}{c|}{\begin{tabular}[c]{@{}c@{}}0.25\\ (0.36)\\ {[}0.13{]}\end{tabular}} & \multicolumn{1}{c|}{\begin{tabular}[c]{@{}c@{}}0.21\\ (0.14)\\ {[}0.02{]}\end{tabular}} & \multicolumn{1}{c|}{\begin{tabular}[c]{@{}c@{}}0.21\\ (0.1)\\ {[}0.01{]}\end{tabular}} & \begin{tabular}[c]{@{}c@{}}0.2\\ (0.05)\\ {[}0.0{]}\end{tabular} &  & \multicolumn{1}{c|}{\begin{tabular}[c]{@{}c@{}}1.67\\ (1.57)\\ {[}2.5{]}\end{tabular}} & \multicolumn{1}{c|}{\begin{tabular}[c]{@{}c@{}}1.59\\ (0.67)\\ {[}0.46{]}\end{tabular}} & \multicolumn{1}{c|}{\begin{tabular}[c]{@{}c@{}}1.35\\ (0.5)\\ {[}0.28{]}\end{tabular}} & \begin{tabular}[c]{@{}c@{}}1.37\\ (0.22)\\ {[}0.07{]}\end{tabular}
    \end{tabular}
    }
    \tablenote{Mean APE estimate after 10000 simulations. Round brackets show the standard deviation and square brackets the MSE of the estimates. $M$ is the order of polynomials in the $Y$-DGP and $N$ the sample size. Within each simulation, the same data is used for all estimators. Details of the DGP and estimated models are shown in Table \ref{table:DGPs} and \ref{table:Models}, respectively. Errors in the $X$-DGP follow a standard normal distribution.}
    \label{table:complex_Y_complex_X}
\end{sidewaystable}

\subsection{Impact of Non Normal Errors}\label{sec:non_normal_errors}
The theoretical conditions on the errors given in Assumption \ref{assump_errors} are satisfied by normally distributed errors with mean zero up to any polynomial degree. To show that results also hold for non-normal errors satisfying the conditions up to a certain order, the following Gaussian mixture is used to generate the errors:
\begin{align*}
P(\nu_{i}) &= 0.5 \cdot N(\mu, 1-\mu^2) + 0.5 \cdot N(-\mu, 1-\mu^2) \quad \text{with} \quad \mu = 0.9
\end{align*}
\begin{equation*}
\vcenter{\hbox{\includegraphics[height=5cm]{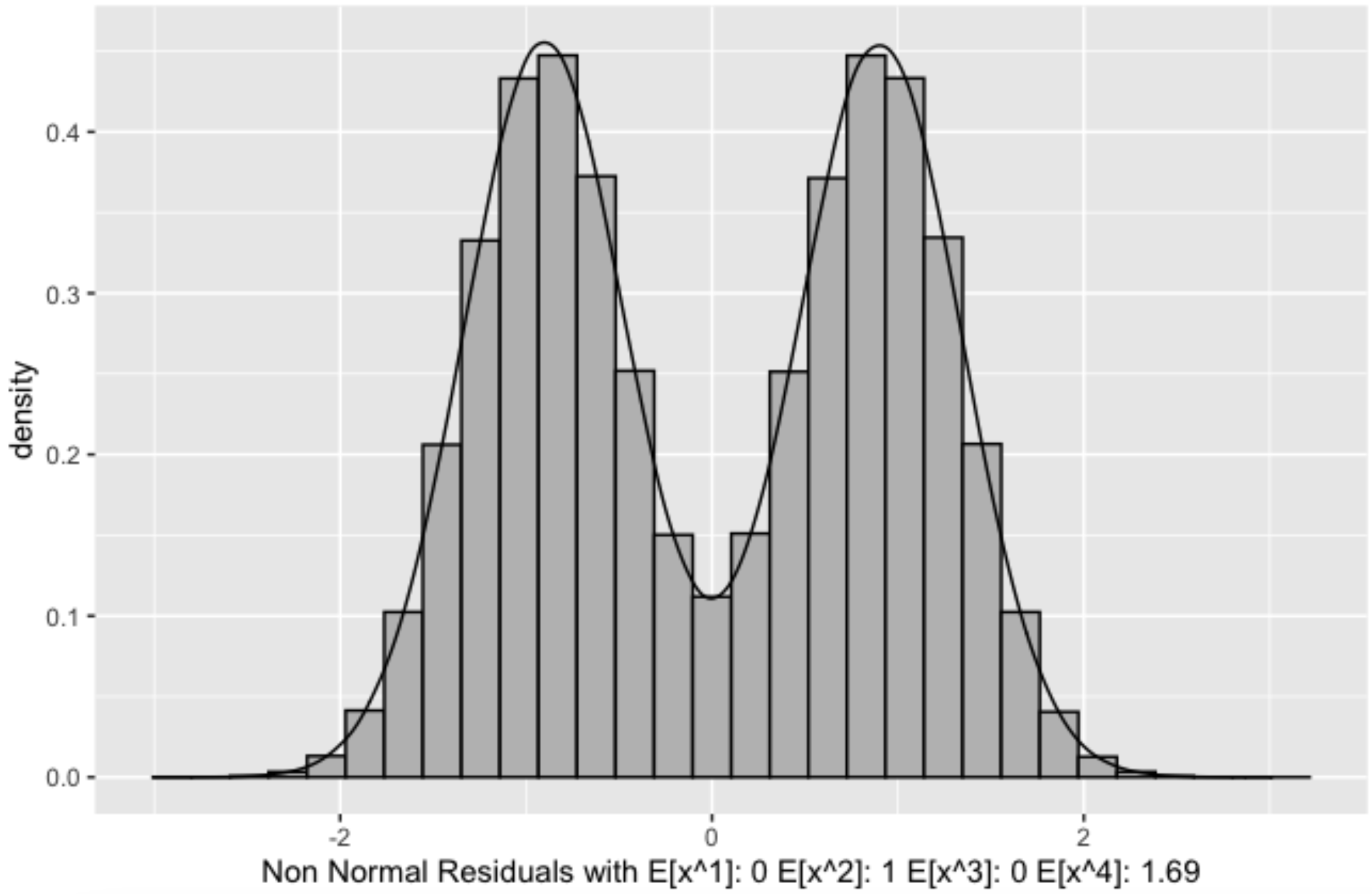}}}
\qquad\qquad
\begin{aligned}
    E\left[ \nu_{i} \right] &= 0 \\
    E\left[ \nu_{i}^2 \right] &= 1 \\
    E\left[ \nu_{i}^3 \right] &= 0 \\
    E\left[ \nu_{i}^4 \right] &\sim 1.7 \neq 3
\end{aligned}
\end{equation*}

These errors do not satisfy the fourth-moment condition in Assumption \ref{assump_errors}, which becomes relevant for a $Y$-DGP with polynomial degree $M\geq3$ in Assumption \ref{assump_DGP}.\footnote{A uniform distribution centered around zero would have the same implications.}

In Table \ref{table:complex_Y_non_normal_error_X}, the results for a simulation with non-normal errors are displayed. In the first two blocks with $M\leq2$ all assumptions of Theorem \ref{thm:main} are satisfied and R-OLS is consistent for the APE. Increasing the polynomial degree in the DGP of $Y$ from $M=2$ to $M=3$ creates a scenario in which Assumption \ref{assump_errors} is not fulfilled anymore.  As suggested by Theorem \ref{thm:main}, R-OLS does not provide an unbiased estimate of the average partial effect anymore. Luckily, the bias is not extreme and the estimates are still less biased than normal OLS or partially linear GAMs. 

\begin{sidewaystable}[]
    \caption{APE estimation with complex $Y$ and complex $X$-DGPs and non normal errors}
    \resizebox{\textwidth}{!}{
    \begin{tabular}{llcccclcccclcccc}
                                 &  & \multicolumn{4}{c}{\textbf{M=1 with APE=0.17}}                                                                                                                                                                                                                                                                                                 &  & \multicolumn{4}{c}{\textbf{M=2 with APE=0.2}}                                                                                                                                                                                                                                                                                                     &  & \multicolumn{4}{c}{\textbf{M=3 with APE=1.52}}                                                                                                                                                                                                                                                                                                   \\ \cline{3-6} \cline{8-11} \cline{13-16} 
                                 &  & \multicolumn{1}{c|}{N=100}                                                              & \multicolumn{1}{c|}{N=500}                                                              & \multicolumn{1}{c|}{N=1000}                                                            & N=5000                                                            &  & \multicolumn{1}{c|}{N=100}                                                              & \multicolumn{1}{c|}{N=500}                                                               & \multicolumn{1}{c|}{N=1000}                                                             & N=5000                                                             &  & \multicolumn{1}{c|}{N=100}                                                              & \multicolumn{1}{c|}{N=500}                                                              & \multicolumn{1}{c|}{N=1000}                                                             & N=5000                                                             \\ \cline{3-6} \cline{8-11} \cline{13-16} 
    \textbf{simple OLS}          &  & \multicolumn{1}{c|}{\begin{tabular}[c]{@{}c@{}}0.13\\ (0.05)\\ {[}0.0{]}\end{tabular}}  & \multicolumn{1}{c|}{\begin{tabular}[c]{@{}c@{}}0.13\\ (0.02)\\ {[}0.0{]}\end{tabular}}  & \multicolumn{1}{c|}{\begin{tabular}[c]{@{}c@{}}0.13\\ (0.02)\\ {[}0.0{]}\end{tabular}} & \begin{tabular}[c]{@{}c@{}}0.13\\ (0.01)\\ {[}0.0{]}\end{tabular} &  & \multicolumn{1}{c|}{\begin{tabular}[c]{@{}c@{}}0.25\\ (0.16)\\ {[}0.03{]}\end{tabular}} & \multicolumn{1}{c|}{\begin{tabular}[c]{@{}c@{}}0.24\\ (0.07)\\ {[}0.01{]}\end{tabular}}  & \multicolumn{1}{c|}{\begin{tabular}[c]{@{}c@{}}0.24\\ (0.05)\\ {[}0.0{]}\end{tabular}}  & \begin{tabular}[c]{@{}c@{}}0.24\\ (0.02)\\ {[}0.0{]}\end{tabular}  &  & \multicolumn{1}{c|}{\begin{tabular}[c]{@{}c@{}}1.01\\ (0.64)\\ {[}0.67{]}\end{tabular}} & \multicolumn{1}{c|}{\begin{tabular}[c]{@{}c@{}}0.98\\ (0.29)\\ {[}0.37{]}\end{tabular}} & \multicolumn{1}{c|}{\begin{tabular}[c]{@{}c@{}}0.98\\ (0.2)\\ {[}0.33{]}\end{tabular}}  & \begin{tabular}[c]{@{}c@{}}0.98\\ (0.09)\\ {[}0.3{]}\end{tabular}  \\ \cline{1-1} \cline{3-6} \cline{8-11} \cline{13-16} 
    \textbf{interacted OLS}      &  & \multicolumn{1}{c|}{\begin{tabular}[c]{@{}c@{}}0.16\\ (0.08)\\ {[}0.01{]}\end{tabular}} & \multicolumn{1}{c|}{\begin{tabular}[c]{@{}c@{}}0.17\\ (0.05)\\ {[}0.01{]}\end{tabular}} & \multicolumn{1}{c|}{\begin{tabular}[c]{@{}c@{}}0.17\\ (0.04)\\ {[}0.0{]}\end{tabular}} & \begin{tabular}[c]{@{}c@{}}0.18\\ (0.03)\\ {[}0.0{]}\end{tabular} &  & \multicolumn{1}{c|}{\begin{tabular}[c]{@{}c@{}}0.32\\ (0.3)\\ {[}0.1{]}\end{tabular}}   & \multicolumn{1}{c|}{\begin{tabular}[c]{@{}c@{}}0.33\\ (0.15)\\ {[}0.04{]}\end{tabular}}  & \multicolumn{1}{c|}{\begin{tabular}[c]{@{}c@{}}0.33\\ (0.12)\\ {[}0.03{]}\end{tabular}} & \begin{tabular}[c]{@{}c@{}}0.34\\ (0.07)\\ {[}0.02{]}\end{tabular} &  & \multicolumn{1}{c|}{\begin{tabular}[c]{@{}c@{}}1.59\\ (1.2)\\ {[}1.44{]}\end{tabular}}  & \multicolumn{1}{c|}{\begin{tabular}[c]{@{}c@{}}1.47\\ (0.55)\\ {[}0.3{]}\end{tabular}}  & \multicolumn{1}{c|}{\begin{tabular}[c]{@{}c@{}}1.5\\ (0.42)\\ {[}0.18{]}\end{tabular}}  & \begin{tabular}[c]{@{}c@{}}1.55\\ (0.22)\\ {[}0.05{]}\end{tabular} \\ \cline{1-1} \cline{3-6} \cline{8-11} \cline{13-16} 
    \textbf{PL-GAM}              &  & \multicolumn{1}{c|}{\begin{tabular}[c]{@{}c@{}}0.14\\ (0.07)\\ {[}0.01{]}\end{tabular}} & \multicolumn{1}{c|}{\begin{tabular}[c]{@{}c@{}}0.12\\ (0.02)\\ {[}0.0{]}\end{tabular}}  & \multicolumn{1}{c|}{\begin{tabular}[c]{@{}c@{}}0.11\\ (0.02)\\ {[}0.0{]}\end{tabular}} & \begin{tabular}[c]{@{}c@{}}0.11\\ (0.01)\\ {[}0.0{]}\end{tabular} &  & \multicolumn{1}{c|}{\begin{tabular}[c]{@{}c@{}}0.19\\ (0.18)\\ {[}0.03{]}\end{tabular}} & \multicolumn{1}{c|}{\begin{tabular}[c]{@{}c@{}}0.15\\ (0.07)\\ {[}0.01{]}\end{tabular}}  & \multicolumn{1}{c|}{\begin{tabular}[c]{@{}c@{}}0.15\\ (0.05\\ {[}0.01{]}\end{tabular}}  & \begin{tabular}[c]{@{}c@{}}0.14\\ (0.02)\\ {[}0.0{]}\end{tabular}  &  & \multicolumn{1}{c|}{\begin{tabular}[c]{@{}c@{}}1.01\\ (0.87)\\ {[}1.02{]}\end{tabular}} & \multicolumn{1}{c|}{\begin{tabular}[c]{@{}c@{}}0.91\\ (0.33)\\ {[}0.48{]}\end{tabular}} & \multicolumn{1}{c|}{\begin{tabular}[c]{@{}c@{}}0.88\\ (0.23)\\ {[}0.46{]}\end{tabular}} & \begin{tabular}[c]{@{}c@{}}0.86\\ (0.1)\\ {[}0.45{]}\end{tabular}  \\ \cline{1-1} \cline{3-6} \cline{8-11} \cline{13-16} 
    \textbf{R-OLS}        &  & \multicolumn{1}{c|}{\begin{tabular}[c]{@{}c@{}}0.15\\ (0.09)\\ {[}0.01{]}\end{tabular}}  & \multicolumn{1}{c|}{\begin{tabular}[c]{@{}c@{}}0.17\\ (0.03)\\ {[}0.0{]}\end{tabular}}  & \multicolumn{1}{c|}{\begin{tabular}[c]{@{}c@{}}0.15\\ (0.03)\\ {[}0.0{]}\end{tabular}} & \begin{tabular}[c]{@{}c@{}}0.13\\ (0.01)\\ {[}0.0{]}\end{tabular} &  & \multicolumn{1}{c|}{\begin{tabular}[c]{@{}c@{}}0.24\\ (0.27)\\ {[}0.08{]}\end{tabular}} & \multicolumn{1}{c|}{\begin{tabular}[c]{@{}c@{}}0.21\\ (0.1)\\ {[}0.01{]}\end{tabular}} & \multicolumn{1}{c|}{\begin{tabular}[c]{@{}c@{}}0.21\\ (0.07)\\ {[}0.01{]}\end{tabular}} & \begin{tabular}[c]{@{}c@{}}0.21\\ (0.03)\\ {[}0.0{]}\end{tabular} &  & \multicolumn{1}{c|}{\begin{tabular}[c]{@{}c@{}}1.53\\ (1.14)\\ {[}1.3{]}\end{tabular}} & \multicolumn{1}{c|}{\begin{tabular}[c]{@{}c@{}}1.32\\ (0.45)\\ {[}0.24{]}\end{tabular}} & \multicolumn{1}{c|}{\begin{tabular}[c]{@{}c@{}}1.27\\ (0.32)\\ {[}0.16{]}\end{tabular}} & \begin{tabular}[c]{@{}c@{}}1.15\\ (0.14)\\ {[}0.16{]}\end{tabular}
    \end{tabular}
    }
    \tablenote{Mean APE estimate after 10000 simulations. Round brackets show the standard deviation and square brackets the MSE of the estimates. $M$ is the order of polynomials in the $Y$-DGP and $N$ the sample size. Within each simulation, the same data is used for all estimators. Details of the DGP and estimated models are shown in Table \ref{table:DGPs} and \ref{table:Models}, respectively. Errors in the $X$-DGP follow a Gaussian mixture with distribution $0.5 \cdot N(\mu, 1-\mu^2) + 0.5 \cdot N(-\mu, 1-\mu^2)$ and $\mu = 0.9$.}
    \label{table:complex_Y_non_normal_error_X}
\end{sidewaystable}

\section{Empirical Illustration} \label{sec:emp_illustrations}
\subsection{Austerity and the Rise of UKIP (Fetzer, 2019)}
This section demonstrates the application of the R-OLS estimator to evaluate the causal impact of austerity measures on political outcomes, leveraging the empirical framework and data from \cite{fetzer2019_austerity_brexit}. The analysis examines the relationship between austerity-induced fiscal losses and the rise in support for the UK Independence Party (UKIP), with broader implications for understanding the Brexit referendum results.

The austerity measures implemented by the UK government after 2010 introduced significant reductions in welfare spending, disproportionately affecting rural and economically disadvantaged regions. These policy changes generated substantial heterogeneity in fiscal shocks, which provides a quasi-experimental setting for causal inference. Austerity exposure is measured as the average financial loss per county, encompassing reductions in tax credits, housing benefits, and disability allowances. This variation allows for an exploration of the political consequences of austerity across districts with differing socioeconomic profiles.

Previous findings by \cite{fetzer2019_austerity_brexit} highlight a robust relationship between austerity exposure and increased support for UKIP. Specifically, districts with higher financial losses experienced greater gains in UKIP vote shares across local, European, and parliamentary elections. These effects were more pronounced in areas with lower educational attainment and higher shares of employment in routine or manufacturing jobs. Using the R-OLS and DML estimator, we refine these estimates by addressing potential non-linear interactions between austerity exposure and observable confounders, providing a more precise estimation of the average partial effect (APE) of austerity on electoral outcomes.

The R-OLS estimation involves two steps. In the first step, we use a Gradient Boosting Machine (GBM) to predict the treatment variable, defined as the interaction term $\mathds{1}(\text{Year} > 2010) \times \text{Austerity}$, where Austerity represents the total financial losses resulting from policy changes. Hyperparameters of the GBM, including the number of trees, interaction depth, minimum observations per node, and shrinkage (learning rate), are tuned using 5-fold cross-validation over a predefined grid. Once the optimal hyperparameters are determined, the data is bootstrapped 250 times. For each bootstrap sample, 5-fold cross-fitting is applied to train the GBM on 4 folds and predict the treatment variable for the left-out fold. The residuals $\hat{\nu}_i$, calculated as the difference between the observed and predicted treatment variable, are then used as inputs in the R-OLS estimator. 

Figure \ref{fig:fetzer_residuals} displays the distribution of residuals after applying the above described residualisation to the entire sample with cross-fitting. It is overlaid with a normal density curve of the same mean and variance. The approximate normality of the residuals provides suggestive evidence that the assumptions under which R-OLS and DML estimate the APE are satisfied.\footnote{The original regression specification in \cite{fetzer2019_austerity_brexit} includes region-by-year fixed effects. Our added controls are: TurnoutPct, populationukmillion, protectiontotalbillion, GVA\_Agriculturepc, GVA\_Manufacturinpc, GVA\_Informationpc, GRetailUKShareWithin, LargeemphigherManAll\_sh, IntermOccAll\_sh, LLTunemployedAll\_sh, LStudentAll\_sh, CMiningAll\_sh, DManufAll\_sh, EUtilityAll\_sh, FConstrAll\_sh, HHotelsAll\_sh, ITransportICTAll\_sh, JFinancialAll\_sh, LPublicAll\_sh, NHealthAll\_sh, OTHERAll\_sh, instrument\_for\_shock, AgeAbove60UK, AC12MigrantShare. See \cite{fetzer2019_austerity_brexit} for details. Without these added controls, the residuals show significant skew as shown in Figure \ref{fig:fetzer_residuals_no_added_controls} in the Appendix and making it unlikely the assumptions of Theorem \ref{thm:main} are satisfied.}

\begin{figure}[h]
    \centering
    \caption{Distribution of residualized treatment variable with additional control variables}
    \includegraphics[width=\textwidth]{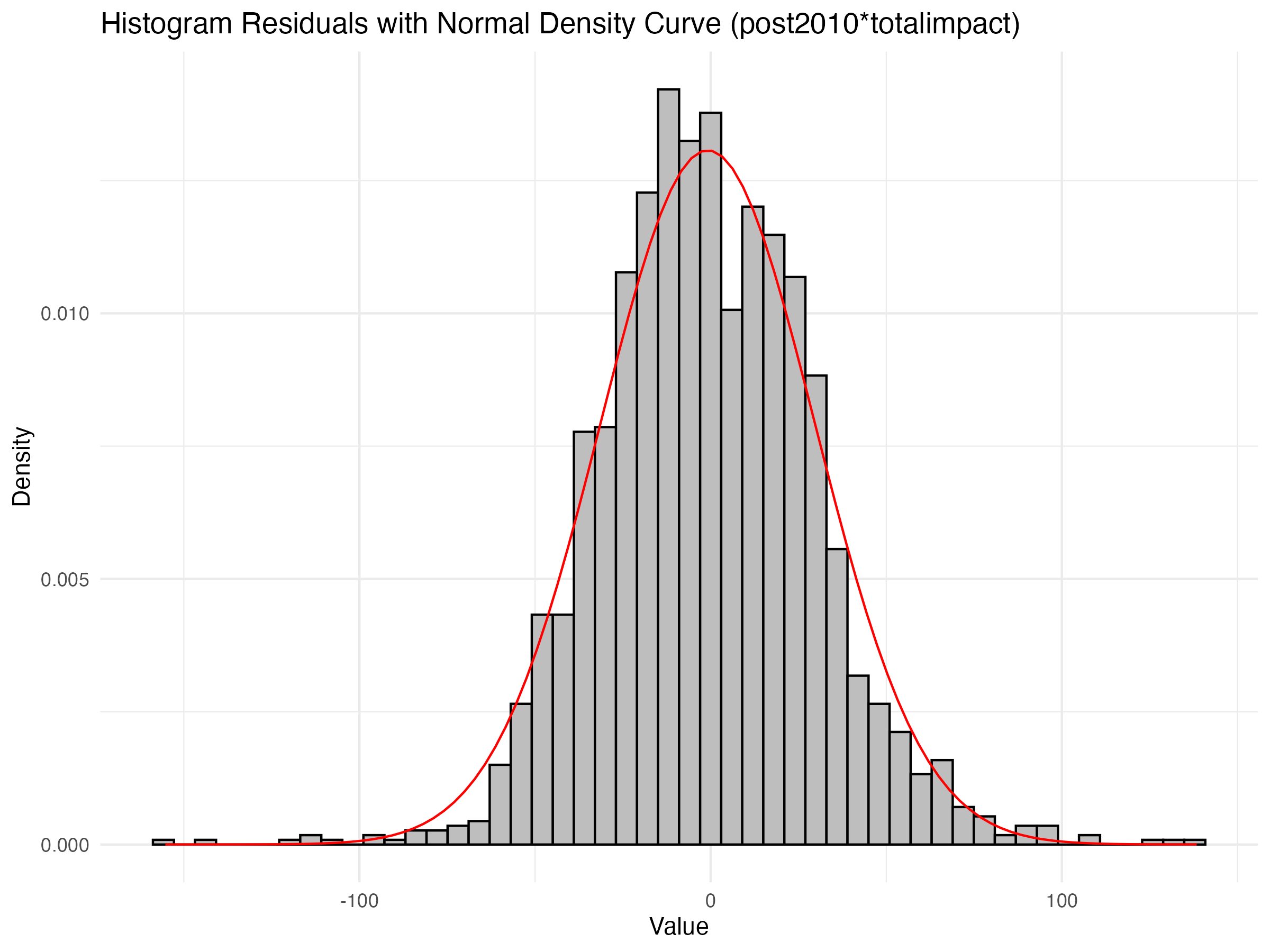}
    \label{fig:fetzer_residuals}
    \tablenote{Histogram of the residuals used in the R-OLS and DML applications to \cite{fetzer2019_austerity_brexit}. The residuals are predicted with a Gradient Boosting Machine (GBM) and overlaid with a standard normal distribution of the same mean and variance. The GBM's hyperparameters are tune via 5-fold cross-validation and 5-fold cross-fitting is applied to predict the treatment of austerity. The included categorical controls are the indicator for post 2010 and region-by-year fixed effects. The added continuous control variables are: TurnoutPct, populationukmillion, protectiontotalbillion, GVA\_Agriculturepc, GVA\_Manufacturinpc, GVA\_Informationpc, GRetailUKShareWithin, LargeemphigherManAll\_sh, IntermOccAll\_sh, LLTunemployedAll\_sh, LStudentAll\_sh, CMiningAll\_sh, DManufAll\_sh, EUtilityAll\_sh, FConstrAll\_sh, HHotelsAll\_sh, ITransportICTAll\_sh, JFinancialAll\_sh, LPublicAll\_sh, NHealthAll\_sh, OTHERAll\_sh, instrument\_for\_shock, AgeAbove60UK, AC12MigrantShare}
\end{figure}

The DML estimator repeats the above outlined residualisation for the two outcome variables, support for UKIP in local and European elections. Residualisation is  performed with a GBM, using the same hyper-parameters as before. The estimator is easily implemented with DoubleML package in R \citep{DoubleML2022Python, DoubleML2022R} using the double machine learning estimator for the partially linear model with the 'partialling out'-score.  Due to the results outlined in Section \ref{sec:DML_ROLS}, we interpret the estimate as an estimate for the APE and use the provided standard errors for inference.\footnote{The DML estimator was included in the bootstrapping framework used for inference for the R-OLS estimator. Estimates of standard error of the APE provided by the bootstrap procedure and the DoubleML package are nearly identical.} 

Both estimators are repeated across all bootstrap samples, yielding 250 estimates of the APE of austerity exposure on UKIP vote shares in local and European elections. The results, summarized in Table \ref{tab:impact}, report the mean and standard deviation of these estimates.

The R-OLS and DML estimates reveal a pronounced sensitivity of UKIP support to austerity-induced financial losses in local elections, exceeding the estimates obtained from conventional fixed-effects regressions (FEOLS) as in \cite{fetzer2019_austerity_brexit}. In contrast, the effect of austerity on European elections is not deemed statistically significant, suggesting that the political consequences of austerity may vary across electoral contexts. These findings underscore the importance of accounting for non-linear treatment effects in empirical research.

\begin{table}[h]
    \caption{Impact of Austerity on Support for UKIP in Local and European Elections}
    \centering
    \begingroup
    \resizebox{\textwidth}{!}{
    \begin{tabular}{lcccc|cccc}
    \hline \hline
    Dependent Variable: & \multicolumn{4}{c}{Local Elections} & \multicolumn{4}{c}{European Elections}\\
    Model: & (FEOLS) & (FEOLS) & (ROLS) & (DML) & (FEOLS) & (FEOLS) & (ROLS) & (DML) \\
    \hline
    $\mathds{1}(Year > 2010) \times \text{Austerity}$ & 0.0163$^{***}$ & 0.0153$^{***}$ & 0.0194$^{**}$ & 0.0151$^{**}$ & 0.0096$^{***}$ & 0.0097$^{***}$ & 0.0103 & 0.0068 \\
    & (0.0030) & (0.0035) & (0.0092) & (0.0061) & (0.0060) & (0.0009) & (0.0129) & 0.0048 \\
    \hline
    \emph{Fixed-effects} &  &  &  &  &  &  &  & \\
    $\mathds{1}(Year > 2010)$ & Yes & Yes & Yes & Yes & Yes & Yes & Yes & Yes \\
    Region $\times$ Year & Yes & Yes & Yes & Yes & Yes & Yes & Yes & Yes \\
    \emph{Additional Controls} & No & Yes & Yes & Yes & No & Yes & Yes & Yes \\
    \hline \hline
    \multicolumn{8}{l}{\emph{Standard-errors in parentheses. Bootstrapped for ROLS.}}\\
    \multicolumn{8}{l}{\emph{Signif. Codes: ***: 0.01, **: 0.05, *: 0.1}}\\
    \end{tabular}
    }
    \endgroup
    \label{tab:impact}
\end{table}

\section{Conclusion} \label{sec:conclusion}
This paper introduces a novel estimation method, R-OLS, for identifying and estimating the average partial effect of a continuous treatment in the presence of non-linear and non-additively separable confounding. The method leverages an exogenous error component of the treatment and imposes specific assumptions on the data-generating process to ensure equivalence to the APE. Double/Debiased Machine Learning is employed to facilitate valid inference for R-OLS estimates. 
Simulation studies demonstrate the strong performance of R-OLS across a range of complex data-generating processes, and applications to empirical data underscore its practical relevance. Overall, R-OLS offers a flexible and robust approach for estimating treatment effects in diverse empirical settings.

\clearpage
\singlespacing
\setlength\bibsep{0pt}
\bibliographystyle{apalike}
\bibliography{bibliography}

\begin{thebibliography}{}

\bibitem[Andrews, 2019]{ANDREWS2019}
Andrews, I. (2019).
\newblock On the structure of iv estimands.
\newblock {\em Journal of Econometrics}, 211(1):294--307.
\newblock Annals Issue in Honor of Jerry A. Hausman.

\bibitem[Angrist, 1998]{Angrist_1998}
Angrist, J.~D. (1998).
\newblock Estimating the labor market impact of voluntary military service using social security data on military applicants.
\newblock {\em Econometrica}, 66(2):249--288.

\bibitem[Angrist and Krueger, 1999]{ANGRIST_KRUEGER_1999}
Angrist, J.~D. and Krueger, A.~B. (1999).
\newblock Chapter 23 - empirical strategies in labor economics.
\newblock volume~3 of {\em Handbook of Labor Economics}, pages 1277--1366. Elsevier.

\bibitem[Bach et~al., 2022]{DoubleML2022Python}
Bach, P., Chernozhukov, V., Kurz, M.~S., and Spindler, M. (2022).
\newblock {DoubleML} -- {A}n object-oriented implementation of double machine learning in {P}ython.
\newblock {\em Journal of Machine Learning Research}, 23(53):1--6.

\bibitem[Bach et~al., 2024]{DoubleML2022R}
Bach, P., Kurz, M.~S., Chernozhukov, V., Spindler, M., and Klaassen, S. (2024).
\newblock {DoubleML}: {A}n object-oriented implementation of double machine learning in {R}.
\newblock {\em Journal of Statistical Software}, 108(3):1--56.
\newblock arXiv:\href{https://arxiv.org/abs/2103.09603}{2103.09603} [stat.ML].

\bibitem[Belloni et~al., 2013]{Belloni_et_al_Lasso_2014}
Belloni, A., Chernozhukov, V., and Hansen, C. (2013).
\newblock Inference on treatment effects after selection among high-dimensional controls.
\newblock {\em The Review of Economic Studies}, 81(2):608--650.

\bibitem[Chernozhukov et~al., 2018]{Chernozhukov_et_al_2018_DebiasedML}
Chernozhukov, V., Chetverikov, D., Demirer, M., Duflo, E., Hansen, C., Newey, W., and Robins, J. (2018).
\newblock Double/debiased machine learning for treatment and structural parameters.
\newblock {\em The Econometrics Journal}, 21(1):C1--C68.

\bibitem[Chernozhukov et~al., 2022]{chernozhukov2022autodml}
Chernozhukov, V., Newey, W.~K., and Singh, R. (2022).
\newblock Automatic debiased machine learning of causal and structural effects.
\newblock {\em Econometrica}, 90(3):967--1027.

\bibitem[Cuesta et~al., 2019]{Cuesta_Steins_Lemma}
Cuesta, J.~I., Davis, J. M.~V., Gianou, A., and Hoyos, A. (2019).
\newblock Identification of average marginal effects under misspecification when covariates are normal.
\newblock {\em Econometric Reviews}, 38(3):350--357.

\bibitem[Fetzer, 2019]{fetzer2019_austerity_brexit}
Fetzer, T. (2019).
\newblock Did austerity cause brexit?
\newblock {\em American Economic Review}, 109(11):3849--86.

\bibitem[Goldsmith-Pinkham et~al., 2024]{goldsmithpinkham_et_al_2022}
Goldsmith-Pinkham, P., Hull, P., and Kolesár, M. (2024).
\newblock Contamination bias in linear regressions.
\newblock {\em American Economic Review}, 114(12):4015–51.

\bibitem[Graham and de~Xavier~Pinto, 2022]{Graham_Pinto_2022}
Graham, B.~S. and de~Xavier~Pinto, C.~C. (2022).
\newblock Semiparametrically efficient estimation of the average linear regression function.
\newblock {\em Journal of Econometrics}, 226(1):115--138.
\newblock Annals Issue in Honor of Gary Chamberlain.

\bibitem[Hansen, 2022]{hansen2022econometrics}
Hansen, B. (2022).
\newblock {\em Econometrics}.
\newblock Princeton University Press.

\bibitem[Hornik et~al., 1989]{HORNIK1989}
Hornik, K., Stinchcombe, M., and White, H. (1989).
\newblock Multilayer feedforward networks are universal approximators.
\newblock {\em Neural Networks}, 2(5):359--366.

\bibitem[Kolesár and Plagborg-Møller, 2024]{kolesar2024dynamiccausaleffectsnonlinearworld}
Kolesár, M. and Plagborg-Møller, M. (2024).
\newblock Dynamic causal effects in a nonlinear world: the good, the bad, and the ugly.

\bibitem[Landsman and Nešlehová, 2008]{LANDSMAN2008}
Landsman, Z. and Nešlehová, J. (2008).
\newblock Stein's lemma for elliptical random vectors.
\newblock {\em Journal of Multivariate Analysis}, 99(5):912--927.

\bibitem[Lee, 2017]{MJ_Lee_2018}
Lee, M.-J. (2017).
\newblock {Simple least squares estimator for treatment effects using propensity score residuals}.
\newblock {\em Biometrika}, 105(1):149--164.

\bibitem[Powell et~al., 1989]{Powell_et_al_1989}
Powell, J.~L., Stock, J.~H., and Stoker, T.~M. (1989).
\newblock Semiparametric estimation of index coefficients.
\newblock {\em Econometrica}, 57(6):1403--1430.

\bibitem[Robins et~al., 1992]{Robins_et_al_1992}
Robins, J.~M., Mark, S.~D., and Newey, W.~K. (1992).
\newblock Estimating exposure effects by modelling the expectation of exposure conditional on confounders.
\newblock {\em Biometrics}, 48(2):479--495.

\bibitem[Robinson, 1988]{Robinson_1988}
Robinson, P.~M. (1988).
\newblock Root-$n$-consistent semiparametric regression.
\newblock {\em Econometrica}, 56(4):931--954.

\bibitem[Ross, 2011]{Ross_Steins_Method}
Ross, N. (2011).
\newblock {Fundamentals of Stein’s method}.
\newblock {\em Probability Surveys}, 8:210 -- 293.

\bibitem[Stein, 1981]{stein1981estimation}
Stein, C.~M. (1981).
\newblock Estimation of the mean of a multivariate normal distribution.
\newblock {\em The Annals of Statistics}, 9(6):1135--1151.

\bibitem[Stoker, 1986]{Stoker_1986}
Stoker, T.~M. (1986).
\newblock Consistent estimation of scaled coefficients.
\newblock {\em Econometrica}, 54(6):1461--1481.

\bibitem[White, 1980]{White_1980_Taylor_Approx}
White, H. (1980).
\newblock Using least squares to approximate unknown regression functions.
\newblock {\em International Economic Review}, 21(1):149--170.

\bibitem[Winkelmann, 2024]{winkelmann_2023}
Winkelmann, R. (2024).
\newblock Neglected heterogeneity, simpson’s paradox, and the anatomy of least squares.
\newblock {\em Journal of Econometric Methods}, 13(1):131--144.

\bibitem[Yitzhaki, 1996]{Yitzhaki1996}
Yitzhaki, S. (1996).
\newblock On using linear regressions in welfare economics.
\newblock {\em Journal of Business \& Economic Statistics}, 14(4):478--486.

\end{thebibliography}

\clearpage

\section*{Appendix A. Proofs} \label{sec:appendix_proof}
\begin{lemma}[Moments of Normal Distribution]
    \label{lemma:normal_moments}
    The moments of $X \sim N(\mu, \sigma^2)$ follow:
    $$ E[X^{p+2}] = \mu E[X^{p+1}] + \sigma^2 (p+1) E[X^p] \quad \text{with} \quad p \in \mathbb{N}$$
\end{lemma}
\begin{proof}[Proof of Lemma \ref{lemma:normal_moments}]
    \label{proof:normal_moments}
    \begin{align*}
        E[X^{p+2}] & =\frac{1}{\sigma \sqrt{2 \pi}} \int x^{p+2} \exp \left(-\frac{(x-\mu)^2}{2 \sigma^2}\right) d x \\
        & =\frac{\sigma^2}{\sigma \sqrt{2 \pi}} \int x^{p+1}\left(\frac{x+\mu-\mu}{\sigma^2}\right) \exp \left(-\frac{(x-\mu)^2}{2 \sigma^2}\right) d x \\
        & =\frac{\mu}{\sigma \sqrt{2 \pi}} \int x^{p+1} \exp \left(-\frac{(x-\mu)^2}{2 \sigma^2}\right) d x+\frac{\sigma^2}{\sigma \sqrt{2 \pi}} \int x^{p+1}\left(\frac{x-\mu}{\sigma^2}\right) \exp \left(-\frac{(x-\mu)^2}{2 \sigma^2}\right) d x \\
        & =\mu E[X^{p+1}]+\frac{\sigma^2}{\sigma \sqrt{2 \pi}} \int x^{p+1}\left(\frac{x-\mu}{\sigma^2}\right) \exp \left(-\frac{(x-\mu)^2}{2 \sigma^2}\right) d x
    \end{align*}
    Applying integration by parts with $u = x^{p+1} \Rightarrow \frac{du}{dx} = (p+1) x^p$ and $v = - \exp \left(-\frac{(x-\mu)^2}{2 \sigma^2}\right) \Rightarrow \frac{dv}{dx} = \frac{x-\mu}{\sigma^2} \exp \left(-\frac{(x-\mu)^2}{2 \sigma^2}\right)$ to the RHS yields:
    \begin{align*}
        E[X^{p+2}] & =\mu E[X^{p+1}]+\frac{\sigma^2}{\sigma \sqrt{2 \pi}} \left[ \left. x^{p+1} \exp \left(-\frac{(x-\mu)^2}{2 \sigma^2}\right) \right|_{-\infty}^{\infty} - \int_{-\infty}^{\infty} (p+1) x^p \exp \left(-\frac{(x-\mu)^2}{2 \sigma^2}\right) dx \right] \\
        & =\mu E[X^{p+1}]+\frac{\sigma^2}{\sigma \sqrt{2 \pi}} \left[ 0 - (p+1) \int_{-\infty}^{\infty} x^p \exp \left(-\frac{(x-\mu)^2}{2 \sigma^2}\right) dx \right] \\
        & =\mu E[X^{p+1}]+\sigma^2 (p+1) E[X^p]
    \end{align*}
    proving the lemma.
\end{proof}

\vspace{1cm}

\newpage

\begin{proof}[Proof Theorem \ref{thm:main}]
    To prove Theorem \ref{thm:main}, we start with the equation for $\beta$ and apply the functional form in Assumption \ref{assump_DGP} and the law of iterated expectation. Then we 
    plug the DGP for $X$: 
    \begin{align*}
        \beta &=  \frac{E\left[\nu_{i} Y_i\right]}{E\left[\nu_{i}^2\right]} 
        = \frac{E\left[\sum_{m=0}^M \nu_{i}X_{i}^mg_m(Z_{i}) \right]}{E\left[\nu_{i}^2\right]} + \frac{E[\nu_{i} \overbrace{E\left[\varepsilon_i \mid X_i, Z_i\right]}^{=0} ]}{E\left[\nu_{i}^2\right]} \\
        &= \frac{1}{E\left[\nu_{i}^2\right]} \sum_{m=0}^M E\left[\nu_{i} (r\left(Z_{i}\right)+\nu_{i})^m g_m(Z_{i})\right] \\
        \shortintertext{Binomial Theorem:}
        &= \frac{1}{E\left[\nu_{i}^2\right]} \sum_{m=0}^M E\left[ \nu_{i} \sum_{p=0}^{m} {m\choose p} r\left(Z_{i}\right)^{m-p} \nu_{i}^p g_m(Z_{i}) \right] \\
        \shortintertext{Independence of $\nu_i$:}
        &= \frac{1}{E\left[\nu_{i}^2\right]} \sum_{m=0}^M \sum_{p=0}^m {m\choose p} E\left[ r\left(Z_{i}\right)^{m-p} g_m\left(Z_{i}\right) \right] E\left[\nu_{i}^{p+1}\right] \\
        \shortintertext{Manipulation of summation and binomial coefficient:}
        &= \frac{1}{E\left[\nu_{i}^2\right]} \sum_{m=0}^M \left[ E\left[ r\left(Z_{i}\right)^{m} g_m\left(Z_{i}\right) \right] E\left[\nu_{i}\right] + \sum_{p=1}^m {m\choose p} E\left[ r\left(Z_{i}\right)^{m-p} g_m\left(Z_{i}\right) \right] E\left[\nu_{i}^{p+1}\right] \right] \\
        &= \frac{1}{E\left[\nu_{i}^2\right]} \sum_{m=1}^M \left[ E\left[ r\left(Z_{i}\right)^{m} g_m\left(Z_{i}\right) \right] E\left[\nu_{i}\right] + \sum_{p=0}^{m-1} {m \choose p+1} E\left[ r\left(Z_{i}\right)^{m-p-1} g_m\left(Z_{i}\right) \right] E\left[\nu_{i}^{p+2}\right] \right] \\
        &= \sum_{m=1}^M \left[ \frac{E\left[ r\left(Z_{i}\right)^{m} g_m\left(Z_{i}\right) \right]}{E\left[\nu_{i}^2\right]} \underbrace{E\left[\nu_{i}\right]}_{=0 \text{, by A\ref{assump_errors}}} + \sum_{p=0}^{m-1} {m-1 \choose p} m E\left[ r\left(Z_{i}\right)^{m-p-1} g_m\left(Z_{i}\right) \right] \underbrace{\frac{E\left[\nu_{i}^{p+2}\right] }{ (p+1) E\left[\nu_{i}^2\right]}}_{=E\left[\nu_{i}^p\right] \text{, by A\ref{assump_errors}}} \right] \\
        &= \sum_{m=1}^M m E\left[ \sum_{p=0}^{m-1} {m-1 \choose p} r\left(Z_{i}\right)^{m-1-p} \nu_{i}^p g_m\left(Z_{i}\right) \right] \\
        \shortintertext{Reverse Binomial Theorem:}
        &= \sum_{m=1}^M E\left[m (r\left(Z_{i}\right)+\nu_{i})^{m-1} g_m\left(Z_{i}\right)\right] \\
        \shortintertext{Definition of X-DGP:}
        &= E\left[\sum_{m=1}^M m X_{i}^{m-1} g_m\left(Z_{i}\right)\right]  \\
        \shortintertext{Definition and Derivative of Y-DGP:}
        &= E\left[\partial_{X_i} Y_i\right]
    \end{align*}
    proving the theorem.
\end{proof}

\begin{proof}[Proof of Lemma \ref{lemma:DML_convergence}]
    We begin the proof by restating the estimator:
    \begin{align*}
        \hat{\theta} &= E_n[(X_i - \hat{r}(Z_i))^2]^{-1} E_n[(X_i - \hat{r}(Z_i)) (Y_i - \hat{l}(Z_i))] \\
        \shortintertext{Substituting the assumption $X_i=r(Z_i)+\nu_i$ and expanding $(Y_i - \hat{l}(Z_i))$ with $\pm l(Z_i)$ yields:}
        &= E_n[(r(Z_i) - \hat{r}(Z_i) + \nu_i)^2]^{-1} E_n[(r(Z_i) - \hat{r}(Z_i) + \nu_i) (Y_i - l(Z_i) + l(Z_i) - \hat{l}(Z_i))] \\
        \shortintertext{Collecting terms gives:}
        &= \Big[ \underbrace{E_n[(r(Z_i) - \hat{r}(Z_i))^2]}_{\xrightarrow{p} 0 \text{ at rate } n^{-2\varphi_r}} + E_n[\nu_i^2] + \underbrace{E_n[2\nu_i(r(Z_i) - \hat{r}(Z_i))]}_{\xrightarrow{p} 0, \text{ with cross-fitting}}\Big]^{-1} \Big[ E_n[\nu_i(Y_i - l(Z_i))] \\ 
        &\quad + \underbrace{E_n[\nu_i(l(Z_i)-\hat{l}(Z_i))]}_{\xrightarrow{p} 0, \text{ with cross-fitting}}  + \underbrace{E_n[(r(Z_i) - \hat{r}(Z_i))(Y_i - l(Z_i))]}_{\xrightarrow{p} 0, \text{ with cross-fitting }}\\
        &\quad + \underbrace{E_n[(r(Z_i) - \hat{r}(Z_i))(l(Z_i) - \hat{l}(Z_i))]}_{\text{Drops if convergence rates satisfy } \sqrt{n}n^{-(\varphi_r + \varphi_l)} \xrightarrow{} 0}\Big] \\
        &= E_n[\nu_i^2]^{-1} E_n[\nu_i(Y_i - l(Z_i))].
    \end{align*}

    In the above derivation, we used cross-fitting heavily to eliminate bias terms. Cross-fitting ensures that nuisance parameter estimates are trained on one part of the data and evaluated on another, mitigating overfitting. For instance, $E_n[\nu_i(r(Z_i) - \hat{r}(Z_i))] = E_n[\nu_i \hat{\nu_i}] \xrightarrow{p} 0$ under cross-fitting because any potential dependence between $\nu_i$ and $\hat{r}(Z_i)$ is broken by independent samples. This logic applies to all terms relying on cross-fitting.

    We also used the assumption that the nuisance function estimators converge fast enough to ensure $\sqrt{n}$-consistency:
    $$
        \sqrt{n}n^{-(\varphi_r + \varphi_l)} \xrightarrow{} 0,
    $$
    implying that the convergence rates satisfy $\varphi_r + \varphi_l > 1/2$. For example, rates faster than $n^{-1/4}$ for both estimators suffice. Residualising the outcome ensures fast convergence of this term and without it stronger assumptions on the convergence rate of machine learning estimators are required. In addition, $E_n[(r(Z_i) - \hat{r}(Z_i))^2]$ converges at rate $n^{-2\varphi_r}$, ensuring any induced attenuation bias vanishes asymptotically. 

    Having eliminated all potential bias terms we can now take the probability limit of the estimator:
    \begin{equation*}
        \hat{\theta} = \frac{E_n[\nu_i(Y_i - l(Z_i))]}{E_n[\nu_i^2]} \xrightarrow{p} \frac{E[\nu_i(Y_i - l(Z_i))]}{E[\nu_i^2]} = \frac{E[\nu_i Y_i]}{E[\nu_i^2]} 
    \end{equation*}
    where we used $E[\nu_i l(Z_i)] = 0$, proving the Lemma.
\end{proof}

\bigskip

\begin{proof}[Proof of Lemma \ref{lemma:DML_APE}]
    Immediate from Lemma \ref{lemma:DML_convergence} and Theorem \ref{thm:main}.
\end{proof}

\bigskip

\begin{lemma}[Linear conditional expectation under joint normality]
    \label{lemma:lin_cond_exp_joint_normal}
    Suppose $Z_i$ and $X_i$ are jointly normal random variables. They satisfy
    $$
        E[Z_i \mid X_i] = E[Z_i] + \rho \frac{\sigma_Z}{\sigma_X} (X_i - E[X_i]) = E[Z_i] + \frac{Cov(X_i, Z_i)}{Var(X_i)} (X_i - E[X_i])
    $$
    where $\rho = \frac{Cov(X_i, Z_i)}{\sigma_X \sigma_Z}$.
\end{lemma}

\begin{proof}[Proof Lemma \ref{lemma:lin_cond_exp_joint_normal}]
    We begin by expanding the conditional expectation:
    \begin{align*}
        E[Z_i \mid X_i] &= E[Z_i - \frac{Cov(X_i, Z_i)}{Var(X_i)} X_i + \frac{Cov(X_i, Z_i)}{Var(X_i)} X_i \mid X_i)  \\
        &= E[Z_i - \frac{Cov(X_i, Z_i)}{Var(X_i)} X_i \mid X_i] + \frac{Cov(X_i, Z_i)}{Var(X_i)} X_i \\
        &= E[Z_i - \frac{Cov(X_i, Z_i)}{Var(X_i)} X_i] + \frac{Cov(X_i, Z_i)}{Var(X_i)} X_i \\
        &= E[Z_i] + \frac{Cov(X_i, Z_i)}{Var(X_i)} (X_i - E[X_i]).
    \end{align*}
    The third equality uses the fact $X_i$ is jointly normal and uncorrelated with $Z_i - \frac{Cov(X_i, Z_i)}{Var(X_i)} X_i$, since both are linear combinations of jointly normal variables which makes them jointly normal again and 
    $$
        Cov(X_i, Z_i - \frac{Cov(X_i, Z_i)}{Var(X_i)} X_i) = Cov(X_i, Z_i) - \frac{Cov(X_i, Z_i)}{Var(X_i)} Var(X_i) = 0.
    $$
    Since joint normality and uncorrelatedness imply independence, we don't need to condition on $X_i$ in the third equality. An alternative proof can be see in Problem 5 in \href{http://athenasc.com/Bivariate-Normal.pdf}{Section 4.7 of the 1st edition (2002) of the book Introduction to Probability, by D. P. Bertsekas and J. N. Tsitsiklis}.
\end{proof}

\bigskip

\begin{proof}[Proof of Lemma \ref{lemma:weights_ROLS}]
    To see which weights underlie Theorem \ref{thm:main}, we start with the general OLS formula, without assuming $E[\nu_i]=0$ and ignoring the error in $Y_i$ for brevity:
    \begin{align*}
        \beta &=  \frac{Cov(\nu_{i}, Y_i)}{Var(\nu_{i})} = \frac{E\left[\nu_{i} Y_i\right] - E\left[\nu_{i}\right] E\left[Y_i\right]}{Var(\nu_{i})} 
        = \frac{E\left[\sum_{m=0}^M \nu_{i}X_{i}^mg_m(Z_{i}) \right] - E[\nu_{i}] E\left[\sum_{m=0}^M X_{i}^mg_m(Z_{i}) \right]}{Var(\nu_i)} \\
        &= \frac{\sum_{m=0}^M E\left[\nu_{i}X_{i}^mg_m(Z_{i}) \right] - E[\nu_{i}] E\left[X_{i}^mg_m(Z_{i}) \right]}{Var(\nu_i)} \\
    \shortintertext{Due to $\nu_i \ind Z_i$ both parts in the sum cancel each other out for $m=0$:}
        &= \frac{1}{Var(\nu_i)} \sum_{m=1}^M E\left[\nu_{i} (r\left(Z_{i}\right)+\nu_{i})^m g_m(Z_{i})\right] - E[\nu_{i}] E\left[(r\left(Z_{i}\right)+\nu_{i})^m g_m(Z_{i})\right]\\
        &= \frac{1}{Var(\nu_i)} \sum_{m=1}^M E\left[ \nu_{i} \sum_{p=0}^{m} {m\choose p} r\left(Z_{i}\right)^{m-p} \nu_{i}^p g_m(Z_{i}) \right] - E[\nu_{i}] E\left[ \sum_{p=0}^{m} {m\choose p} r\left(Z_{i}\right)^{m-p} \nu_{i}^p g_m(Z_{i}) \right] \\
        &= \frac{1}{Var(\nu_i)} \sum_{m=1}^M \sum_{p=0}^m {m\choose p} E\left[ r\left(Z_{i}\right)^{m-p} g_m\left(Z_{i}\right) \right] E\left[\nu_{i}^{p+1}\right] - E[\nu_{i}] E\left[ r\left(Z_{i}\right)^{m-p} g_m\left(Z_{i}\right) \right] E\left[\nu_{i}^{p}\right] \\
        &= \frac{1}{Var(\nu_i)} \sum_{m=1}^M \sum_{p=0}^m {m\choose p} E\left[ r\left(Z_{i}\right)^{m-p} g_m\left(Z_{i}\right) \right] \left[ E\left[\nu_{i}^{p+1}\right] - E[\nu_{i}] E\left[\nu_{i}^{p}\right] \right] \\
        &= \frac{1}{Var(\nu_i)} \sum_{m=1}^M \sum_{p=0}^{m-1} {m \choose p+1} E\left[ r\left(Z_{i}\right)^{m-p-1} g_m\left(Z_{i}\right) \right] \left[ E\left[\nu_{i}^{p+2}\right] - E[\nu_{i}] E\left[\nu_{i}^{p+1}\right] \right] \\
        &= \frac{1}{Var(\nu_i)} \sum_{m=1}^M \sum_{p=0}^{m-1} {m-1 \choose p} \frac{m}{p+1} E\left[ r\left(Z_{i}\right)^{m-p-1} g_m\left(Z_{i}\right) \right] \left[ E\left[\nu_{i}^{p+2}\right] - E[\nu_{i}] E\left[\nu_{i}^{p+1}\right] \right] \\
        &= \sum_{m=1}^M \sum_{p=0}^{m-1} {m-1 \choose p} m E\left[ r\left(Z_{i}\right)^{m-p-1} g_m\left(Z_{i}\right) \right] \frac{E\left[\nu_{i}^{p+2}\right] - E[\nu_{i}] E\left[\nu_{i}^{p+1}\right]}{ (p+1) Var(\nu_i)} \\
        &= \sum_{m=1}^M \sum_{p=0}^{m-1} {m-1 \choose p} m E\left[ r\left(Z_{i}\right)^{m-p-1} g_m\left(Z_{i}\right) \right] \frac{E\left[\nu_{i}^{p+2}\right] - E[\nu_{i}] E\left[\nu_{i}^{p+1}\right]}{ (p+1) Var(\nu_i)} \frac{E\left[\nu_{i}^{p}\right]}{E\left[\nu_{i}^{p}\right]} \\
        &= \sum_{m=1}^M \sum_{p=0}^{m-1} {m-1 \choose p} m E\left[ r\left(Z_{i}\right)^{m-p-1} g_m\left(Z_{i}\right) \right] E\left[\nu_{i}^{p}\right] \frac{E\left[\nu_{i}^{p+2}\right] - E[\nu_{i}] E\left[\nu_{i}^{p+1}\right]}{ (p+1) Var(\nu_i) E\left[\nu_{i}^{p}\right]} \\
        &= \sum_{m=1}^M \sum_{p=0}^{m-1} {m-1 \choose p} \underbrace{m E\left[ r\left(Z_{i}\right)^{m-1-p} \nu_{i}^{p} g_m\left(Z_{i}\right) \right]}_{\text{APE of the p-th element of the m-th order polynomial}} \underbrace{\frac{E\left[\nu_{i}^{p+2}\right] - E[\nu_{i}] E\left[\nu_{i}^{p+1}\right]}{ (p+1) Var(\nu_i) E\left[\nu_{i}^{p}\right]}}_{\text{Weight}} 
    \end{align*}
    showing the weights implicitly applied to the APEs by R-OLS.
\end{proof}

\bigskip

\begin{lemma}[APE in Interacted OLS]
    \label{lemma:variance_APE_interacted_OLS}
    Assume $Y=g(X,Z) + \varepsilon$ with $g : \mathbb{R}^K \rightarrow \mathbb{R}$, $X \in \mathbb{R}$ and $Z \in \mathbb{R}^K-1$ follows a known polynomial regression with interactions between covariates. 
    Then the average partial effect of $X$ on $Y$ is given by:
    \begin{equation*}
        \widehat{APE}_{X} = \frac{1}{N} \sum_{i=1}^N \frac{\partial}{\partial X} \hat{g}(x_i, z_i)
    \end{equation*}
    where $\hat{g}(X,Z)$ is the OLS estimator of $g(X,Z)$.
    The variance of the average partial effect estimate is given by the Delta-Method:
    \begin{equation*}
        Var\left(\widehat{APE}_{X} \mid X, Z\right) = \left( \nabla_{\beta} \widehat{APE}_{X} \right)^T Var\left(\hat{\beta} \mid X, Z\right) \left( \nabla_{\beta} \widehat{APE}_{X} \right)
    \end{equation*}   
    where $\hat{\beta}$ is the vector of coefficients in the OLS regression of $Y$ on all interactions and polynomials of $X$ and $Z$.
\end{lemma}

\section*{Appendix B. Simulation Results} \label{sec:appendix_simulation}
\begin{sidewaystable}
    \caption{APE estimation with complex $Y$ and additive $X$-DGPs and standard normal errors}
    \resizebox{\textwidth}{!}{
    \begin{tabular}{llcccclcccclcccc}
                                    &  & \multicolumn{4}{c}{\textbf{M=1 with APE=0.17}}                                                                                                                                                                                                                                                                                                       &  & \multicolumn{4}{c}{\textbf{M=2 with APE=0.53}}                                                                                                                                                                                                                                                                                                      &  & \multicolumn{4}{c}{\textbf{M=3 with APE=0.19}}                                                                                                                                                                                                                                                                                                            \\ \cline{3-6} \cline{8-11} \cline{13-16} 
                                    &  & \multicolumn{1}{c|}{N=100}                                                               & \multicolumn{1}{c|}{N=500}                                                               & \multicolumn{1}{c|}{N=1000}                                                              & N=5000                                                              &  & \multicolumn{1}{c|}{N=100}                                                               & \multicolumn{1}{c|}{N=500}                                                               & \multicolumn{1}{c|}{N=1000}                                                              & N=5000                                                              &  & \multicolumn{1}{c|}{N=100}                                                                  & \multicolumn{1}{c|}{N=500}                                                                & \multicolumn{1}{c|}{N=1000}                                                               & N=5000                                                               \\ \cline{3-6} \cline{8-11} \cline{13-16} 
    \textbf{simple OLS}          &  & \multicolumn{1}{c|}{\begin{tabular}[c]{@{}c@{}}0.17\\ (0.15)\\ {[}0.02{]}\end{tabular}} & \multicolumn{1}{c|}{\begin{tabular}[c]{@{}c@{}}0.17\\ (0.07)\\ {[}0.0{]}\end{tabular}}  & \multicolumn{1}{c|}{\begin{tabular}[c]{@{}c@{}}0.17\\ (0.05)\\ {[}0.0{]}\end{tabular}}    & \begin{tabular}[c]{@{}c@{}}0.17\\ (0.02)\\ {[}0.0{]}\end{tabular}     &  & \multicolumn{1}{c|}{\begin{tabular}[c]{@{}c@{}}0.54\\ (0.77)\\ {[}0.59{]}\end{tabular}}    & \multicolumn{1}{c|}{\begin{tabular}[c]{@{}c@{}}0.53\\ (0.34)\\ {[}0.12{]}\end{tabular}} & \multicolumn{1}{c|}{\begin{tabular}[c]{@{}c@{}}0.53\\ (0.24)\\ {[}0.06{]}\end{tabular}} & \begin{tabular}[c]{@{}c@{}}0.53\\ (0.11)\\ {[}0.01{]}\end{tabular} &  & \multicolumn{1}{c|}{\begin{tabular}[c]{@{}c@{}}0.24\\ (4.39)\\ {[}19.26{]}\end{tabular}}  & \multicolumn{1}{c|}{\begin{tabular}[c]{@{}c@{}}0.22\\ (1.95)\\ {[}3.81{]}\end{tabular}}   & \multicolumn{1}{c|}{\begin{tabular}[c]{@{}c@{}}0.24\\ (1.39)\\ {[}1.93{]}\end{tabular}}   & \begin{tabular}[c]{@{}c@{}}0.20\\ (0.62)\\ {[}0.39{]}\end{tabular} \\ \cline{1-1} \cline{3-6} \cline{8-11} \cline{13-16} 
    \textbf{interacted OLS}      &  & \multicolumn{1}{c|}{\begin{tabular}[c]{@{}c@{}}0.16\\ (0.11)\\ {[}0.01{]}\end{tabular}} & \multicolumn{1}{c|}{\begin{tabular}[c]{@{}c@{}}0.16\\ (0.05)\\ {[}0.0{]}\end{tabular}}  & \multicolumn{1}{c|}{\begin{tabular}[c]{@{}c@{}}0.17\\ (0.03)\\ {[}0.0{]}\end{tabular}}    & \begin{tabular}[c]{@{}c@{}}0.17\\ (0.01)\\ {[}0.0{]}\end{tabular}     &  & \multicolumn{1}{c|}{\begin{tabular}[c]{@{}c@{}}0.46\\ (0.57)\\ {[}0.32{]}\end{tabular}}    & \multicolumn{1}{c|}{\begin{tabular}[c]{@{}c@{}}0.51\\ (0.23)\\ {[}0.05{]}\end{tabular}} & \multicolumn{1}{c|}{\begin{tabular}[c]{@{}c@{}}0.52\\ (0.16)\\ {[}0.03{]}\end{tabular}} & \begin{tabular}[c]{@{}c@{}}0.53\\ (0.07)\\ {[}0.01{]}\end{tabular} &  & \multicolumn{1}{c|}{\begin{tabular}[c]{@{}c@{}}-0.09\\ (3.28)\\ {[}10.83{]}\end{tabular}} & \multicolumn{1}{c|}{\begin{tabular}[c]{@{}c@{}}0.10\\ (1.40)\\ {[}1.98{]}\end{tabular}}   & \multicolumn{1}{c|}{\begin{tabular}[c]{@{}c@{}}0.16\\ (1.00)\\ {[}1.00{]}\end{tabular}}   & \begin{tabular}[c]{@{}c@{}}0.18\\ (0.45)\\ {[}0.21{]}\end{tabular}  \\ \cline{1-1} \cline{3-6} \cline{8-11} \cline{13-16} 
    \textbf{PL-GAM}              &  & \multicolumn{1}{c|}{\begin{tabular}[c]{@{}c@{}}0.19\\ (0.15)\\ {[}0.02{]}\end{tabular}} & \multicolumn{1}{c|}{\begin{tabular}[c]{@{}c@{}}0.17\\ (0.06)\\ {[}0.0{]}\end{tabular}}  & \multicolumn{1}{c|}{\begin{tabular}[c]{@{}c@{}}0.17\\ (0.04)\\ {[}0.0{]}\end{tabular}}    & \begin{tabular}[c]{@{}c@{}}0.17\\ (0.02)\\ {[}0.0{]}\end{tabular}     &  & \multicolumn{1}{c|}{\begin{tabular}[c]{@{}c@{}}0.66\\ (0.74)\\ {[}0.57{]}\end{tabular}}    & \multicolumn{1}{c|}{\begin{tabular}[c]{@{}c@{}}0.55\\ (0.30)\\ {[}0.09{]}\end{tabular}} & \multicolumn{1}{c|}{\begin{tabular}[c]{@{}c@{}}0.54\\ (0.22)\\ {[}0.05{]}\end{tabular}} & \begin{tabular}[c]{@{}c@{}}0.53\\ (0.10)\\ {[}0.01{]}\end{tabular} &  & \multicolumn{1}{c|}{\begin{tabular}[c]{@{}c@{}}1.02\\ (4.14)\\ {[}17.79{]}\end{tabular}}  & \multicolumn{1}{c|}{\begin{tabular}[c]{@{}c@{}}0.36\\ (1.77)\\ {[}3.15{]}\end{tabular}}   & \multicolumn{1}{c|}{\begin{tabular}[c]{@{}c@{}}0.29\\ (1.27)\\ {[}1.61{]}\end{tabular}}   & \begin{tabular}[c]{@{}c@{}}0.21\\ (0.57)\\ {[}0.33{]}\end{tabular}  \\ \cline{1-1} \cline{3-6} \cline{8-11} \cline{13-16} 
    \textbf{R-OLS}               &  & \multicolumn{1}{c|}{\begin{tabular}[c]{@{}c@{}}0.18\\ (0.13)\\ {[}0.02{]}\end{tabular}} & \multicolumn{1}{c|}{\begin{tabular}[c]{@{}c@{}}0.17\\ (0.06)\\ {[}0.0{]}\end{tabular}}  & \multicolumn{1}{c|}{\begin{tabular}[c]{@{}c@{}}0.17\\ (0.03)\\ {[}0.0{]}\end{tabular}}    & \begin{tabular}[c]{@{}c@{}}0.17\\ (0.02)\\ {[}0.0{]}\end{tabular}     &  & \multicolumn{1}{c|}{\begin{tabular}[c]{@{}c@{}}0.60\\ (0.67)\\ {[}0.45{]}\end{tabular}}    & \multicolumn{1}{c|}{\begin{tabular}[c]{@{}c@{}}0.54\\ (0.30)\\ {[}0.09{]}\end{tabular}} & \multicolumn{1}{c|}{\begin{tabular}[c]{@{}c@{}}0.58\\ (0.20)\\ {[}0.04{]}\end{tabular}} & \begin{tabular}[c]{@{}c@{}}0.53\\ (0.09)\\ {[}0.01{]}\end{tabular} &  & \multicolumn{1}{c|}{\begin{tabular}[c]{@{}c@{}}0.41\\ (4.17)\\ {[}17.45{]}\end{tabular}}  & \multicolumn{1}{c|}{\begin{tabular}[c]{@{}c@{}}0.23\\ (1.95)\\ {[}3.79{]}\end{tabular}}   & \multicolumn{1}{c|}{\begin{tabular}[c]{@{}c@{}}0.33\\ (1.22)\\ {[}1.50{]}\end{tabular}}   & \begin{tabular}[c]{@{}c@{}}0.20\\ (0.55)\\ {[}0.30{]}\end{tabular} 
    \end{tabular}
    }
    \tablenote{Mean APE estimate after 10000 simulations. Round brackets show the standard deviation and square brackets the MSE of the estimates. $M$ is the order of polynomials in the $Y$-DGP and $N$ the sample size. Within each simulation, the same data is used for all estimators. Details of the DGP and estimated models are shown in Table \ref{table:DGPs} and \ref{table:Models}, respectively. Errors in the $X$-DGP follow a standard normal distribution.}
    \label{table:complex_Y_additive_X}
\end{sidewaystable}

\newpage

\section*{Appendix: Empirical Illustration}
\begin{figure}[h]
    \centering
    \caption{Distribution of residualized treatment variable without additional control variables}
    \includegraphics[width=0.7\textwidth]{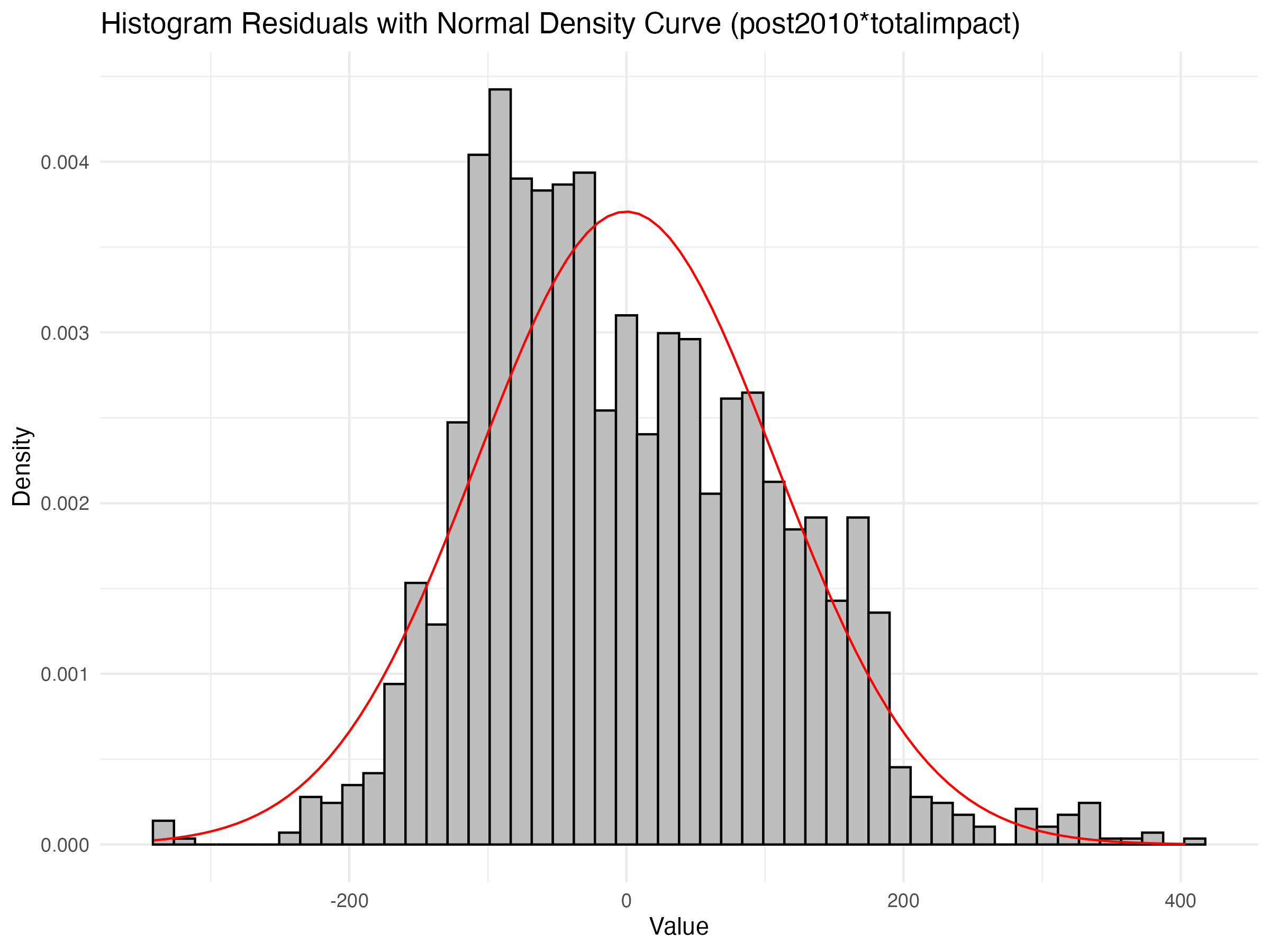}
    \label{fig:fetzer_residuals_no_added_controls}
    \tablenote{Histogram of the residuals used in the R-OLS and DML applications to \cite{fetzer2019_austerity_brexit}. The residuals are predicted with a Gradient Boosting Machine (GBM) and overlaid with a standard normal distribution of the same mean and variance. The GBM's hyperparameters are tune via 5-fold cross-validation and 5-fold cross-fitting is applied to predict the treatment of austerity. The included categorical controls are the indicator for post 2010 and region-by-year fixed effects. No continuous control variables have been added}
\end{figure}

\newpage

\section*{Appendix C. Mesokurtotic Distributions} \label{sec:appendix_mesokurtotic}
Examples of mesokurtotic distributions other than the normal distribution can be created by combining two distributions in a mixture. The resulting distribution is mesokurtotic if the kurtosis of the mixture equals three. Some of the following examples are taken from \href{https://stats.stackexchange.com/a/154965/254653}{https://stats.stackexchange.com/a/154965/254653}.
\subsection*{Mixture of two uniform distributions}
Combine two uniform distributions, where $U_1 \sim U(-1,1)$ and $U_2 \sim U(-a,a)$ and call the mixture $M = 0.5*U_1 + 0.5*U_2$. Due to the symmetry $E[M]=0$. The variance of $M$ is given by 
$$Var[M] = E[M^2] = 0.5E[U_1^2] + 0.5E[U_2^2] = 0.5\frac{1}{3} + 0.5^2\frac{a^2}{3} = \frac{1}{6} + \frac{a^2}{6} = \frac{1+a^2}{6}$$
Similarly, $$E[M^4] = 0.5E[U_1^4] + 0.5E[U_2^4] = 0.5*(\frac{1}{5} + \frac{a^4}{5}) = \frac{1}{10}(1+a^4).$$
Hence, the kurtosis is given by 
$$ \frac{\frac{1}{10}(1+a^4)}{(\frac{1+a^2}{6})^2} = 3.6 \frac{1 + a^4}{(1+a^2)^2} $$
which equals three if $a=\sqrt{5 + \sqrt{24}}$. The distribution is shown in Figure \ref{fig:mesokurtotic_distributions}.
\begin{figure}[H]
    \centering
    \caption{Mesokurtotic Uniform Mixture}
    \includegraphics[width=0.7\textwidth]{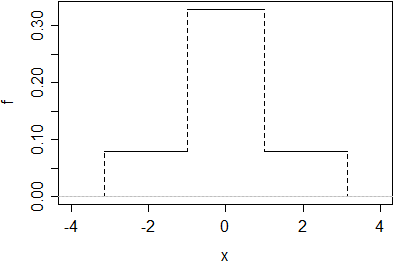}
    \label{fig:mesokurtotic_distributions}
\end{figure}

\section*{Appendix D: OLS and IV Weights} \label{sec:appendix_OLS_IV_weights}
Intuition dictates that the OLS estimator is a weighted average of the APEs of a polynomial regression. We provide a novel result, decomposing the OLS estimand as a weighted sum of the APEs of each element in the polynomial DGP of $Y_i=\sum_{m=0}^M X_{i}^m \beta_m$. Focussing on the univariate case with $\beta_m$ as a scalar:
\begin{align*}
    \beta_{OLS} &= \frac{Cov(X_i, Y_i)}{Var(X_i)} \\ 
    & = \frac{E[X_i \sum_{m=0}^{M} X_i^m \beta_m]}{E[X_i(X_i-E[X_i])]} - \frac{E[X_i] E[\sum_{m=0}^{M} X_i^m \beta_m]}{E[X_i(X_i-E[X_i])]} \\
    &= \sum_{m=1}^{M} \frac{(E[X_i^{m+1}] - E[X_i] E[X_i^{m}])}{E[X_i(X_i-E[X_i])]} \beta_m \\ 
    &= \sum_{m=1}^{M} \frac{(E[X_i^{m+1}] - E[X_i] E[X_i^{m}])}{m E[X_i(X_i-E[X_i])] E[X_i^{m-1}]} E[m X_i^{m-1} \beta_m] \numberthis \label{eq:OLS_weights_APE} 
\end{align*}
where we expanded the expression with $\frac{E[m X_i^{m-1}]}{E[m X_i^{m-1}]}$ in the last equality.
This result shows that OLS effectively estimates the true polynomial structure of the Y-DGP, calculates the APE of each polynomial and weights it according to a ratio of moments of $X_i$ given by:
$$
    \frac{(E[X_i^{m+1}] - E[X_i] E[X_i^{m}])}{m \underbrace{E[X_i(X_i-E[X_i])]}_{=Var(X_i)} E[X_i^{m-1}]}
$$

If the true Y-DGP is not a polynomial in $X_i$, this decomposition shows that OLS does a Taylor expansion of infinite order around 0, calculates the average partial effect for each "Taylor-Polynomial" and applies the above given weights.\footnote{Taylor expansions around the mean would imply slightly different weights.} This is in contrast to the equally valid Taylor expansion interpretations given in \cite{White_1980_Taylor_Approx} and \cite{kolesar2024dynamiccausaleffectsnonlinearworld}.

We can see that this results is similar to the results for R-OLS in Lemma \ref{lemma:weights_ROLS}. The main difference is that R-OLS splits the APEs further into p elements, weighting the partial derivative of $Y_i$ wrt $\nu_i$. The advantage of this is, that the distribution of $X_i$ can be largely unspecified, and distributional assumption only need to me made for the error term in $X_i$. The disadvantage is that all moments of $\nu_i$ below $m$ need to satisfy assumptions such that the weight applied to each APE equals one, while equation (\ref{eq:OLS_weights_APE}) only requires the moments corresponding to specific $m$'s in the Y-DGP to satisfy these moments.

For distributions of $X_i$ with $E[X_i]=0$, the connection between the moments dictates that $X_i$ is symmetric and odd moments of the included polynomials are 0, giving the expression:
$$
    \beta_{OLS} = \sum_{m=1 \text{ and m is odd}}^{M} \frac{E[X_i^{m+1}]}{m E[X_i^2] E[X_i^{m-1}]} E[m X_i^{m-1} \beta_m]
$$
where only the odd-order elements of the polynomial matter, since the derivative of the even moments are 0. The true polynomial can still include even powered polynomials, their APE is just always zero under symmetry in $X_i$.

Alternative decompositions of OLS provided by \cite{Yitzhaki1996, Angrist_1998, ANGRIST_KRUEGER_1999, Graham_Pinto_2022, kolesar2024dynamiccausaleffectsnonlinearworld} show:
\begin{align*}
    \beta_{OLS} = \frac{Cov(X_i, Y_i)}{Var(X_i)} = E\left[ \frac{\omega(x_i)}{E[\omega(x_i)]} \partial_X g(x_i) \right] = \sum_{m=1}^M E\left[ \frac{\omega(x_i)}{E[\omega(x_i)]} m x_i^{m-1} \beta_m \right]
\end{align*}
where:
\begin{align*}
    \omega(x_i) = \frac{1}{f_X(x_i)} \left[E[X_i \mid X_i \geq x_i] - E[X_i \mid X_i < x_i]\right] Pr(X_i \geq x_i) Pr(X_i < x_i).
\end{align*}

Since both weighting schemes need to agree, we know that once the distribution of $X_i$ satisfies Assumption \ref{assump_errors}, the weights $\omega(x_i)$ need to be orthogonal to the derivative $\partial_X g(x_i)$, under the distribution of $X_i$ and the polynomial assumption for Y, providing an alternative interpretation of the assumptions used in R-OLS.

We can do a similar derivation for the IV estimator. The IV estimator is given by:
\begin{align*}
    \beta &= \frac{Cov(W_i, Y_i)}{Cov(W_i, X_i)} = \frac{E[W_i Y_i] - E[W_i] E[Y_i]}{Cov(W_i, X_i)} \\
    & = \frac{E[W_i \sum_{m=0}^{M} X_i^m \beta_m + \zeta_i] - E[W_i] E[\sum_{m=0}^{M} X_i^m \beta_m + \zeta_i]}{Cov(W_i, X_i)} \\
    & = \sum_{m=1}^{M} \frac{\left[ E[W_i X_i^{m}] - E[W_i] E[X_i^{m}] \right] \beta_m}{Cov(W_i, X_i)} \\
    & = \sum_{m=1}^{M} \frac{\left[ E[W_i X_i^{m}] - E[W_i] E[X_i^{m}] \right] \beta_m}{Cov(W_i, X_i)} \frac{m E[X_i^{m-1}]}{m E[X_i^{m-1}]} \\
    & = \sum_{m=1}^{M} \underbrace{\frac{E[W_i X_i^{m}] - E[W_i] E[X_i^{m}]}{m Cov(W_i, X_i) E[X_i^{m-1}]}}_{\text{Weight}} \underbrace{E[m X_i^{m-1} \beta_m]}_{\text{APE of the m-th element in the polynomial}}
\end{align*}
Subject to the implied independence assumptions, especially with $\beta_m$ and $\zeta_i$.

If these weights equal 1, then IV estimates the APE. However, if the DGP for $X_i$ is a polynomial in $W_i$ of order 2 or higher, this becomes increasingly unlikely due to the double exponential structure putting substantially stricter conditions on the distribution of $W_i$. 

For example, assume $X_i=W_i + W_i^2 + \zeta_i$, then the weight becomes:
\begin{align*}
    & \frac{E[W_i (W_i + W_i^2 + \zeta_i)^{m}] - E[W_i] E[(W_i + W_i^2 + \zeta_i)^{m}]}{m Cov(W_i, W_i + W_i^2 + \zeta_i) E[(W_i + W_i^2 + \zeta_i)^{m-1}]} \\
    =& \frac{E[W_i^{2*m + 1} + ...] - E[W_i] E[W_i^{2*m} + ...]}{m \left[ Var(W_i) + Cov(W_i, W_i^2) \right] E[W_i^{2*(m-1)} + ... ]}
\end{align*}
placing very restrictive conditions on the distribution of $W_i$.

\newpage
\section*{Appendix E: Identification of the APE with Instrumental Variables} \label{sec:IV}


\cite{Cuesta_Steins_Lemma} study the IV estimand under heterogeneous treatment effects with additively separable endogeneity. They use the bivariate version of Stein's Lemma to show that the instrumental variables estimator estimates the average partial effect of $X_i$ on $g(X_i)$
$$
    \frac{Cov(W_i, g(X_i))}{Cov(X_i, W_i)} = E\left[ \partial_{X_i} Y_i\right],
$$
if the endogeneous $X_i$ and the instrument $W_i$ are jointly normally distributed. As Lemma \ref{lemma:lin_cond_exp_joint_normal} in the Appendix shows, joint normality is restrictive since it is only satisfied under linear conditional expectations, meaning the relationship between $X_i$ and $W_i$ needs to be linear with additively separable noise.\footnote{It is possible to relax the additive separability of the noise, but then one needs joint normality of $X_i$, $W_i$ and $\varepsilon_i$, as shown by \cite{Cuesta_Steins_Lemma}}

We extend the R-OLS framework to IV regressions to show under which conditions the IV estimand coincides with the average partial effect (APE) and how heterogeneity influences identification. We discuss two forms of heterogeneity: (1) heterogeneity in the treatment effect of $X_i$ on $Y_i$, as in the R-OLS case; and (2) heterogeneity in the relationship between $W_i$ and $X_i$, which is absent in R-OLS and assumed to be linear in \cite{Cuesta_Steins_Lemma}. We demonstrate that these two forms of heterogeneity are inherently connected and guarantee the identification of the APE only under strong assumptions, joint normality being one of the few exceptions satisfying the assumptions.

Consider the following data-generating process for IV estimation with heterogeneous treatment effects:
\begin{assumption}[Endogenous Treatment with Heterogeneous Effects]
    \label{assump:IV_DGP} 
    Let $(Y_i, X_i, W_i, Z_i)$ be independently and identically distributed random vectors satisfying: 
    \begin{align}
        &Y_i = \sum_{m=0}^M X_{i}^mg_m(Z_i) + \varepsilon_i \quad \text{for} \quad M \in \mathbb{N} \label{DGP:Y_IV} \\
        &X_{i} = r\left(W_i, Z_i\right)+\zeta_{i} \label{DGP:X_IV}
    \end{align}
    with $X_i \in \mathbb{R}$, $W_i \in \mathbb{R}$, $Z_i \in \mathbb{R}^K$, $E[W_i] = 0$, $E[W_i \varepsilon_i] = 0$, $E[\varepsilon_i \zeta_i]\neq0$, $E[|X_i W_i|] < \infty$ and $E[|W_iY_i|] < \infty$.
\end{assumption}

This framework generalizes the classical IV setup by accommodating heterogeneity. Equation (\ref{DGP:Y_IV}) allows for non-linearities and interactions in the treatment effect of $X_i$ on $Y_i$, mirroring the R-OLS case. Equation (\ref{DGP:X_IV}) introduces heterogeneity in the effect of $W_i$ on $X_i$, encompassing both non-linearities and potential confounding by $Z_i$. Endogeneity of $X_i$ arises from the correlation between $\varepsilon_i$ and $\zeta_i$. Combined with the exclusion restriction, $E[W_i \varepsilon_i] = 0$, this implies $E[W_i \zeta_i] = 0$, since any non-zero $E[W_i \zeta_i]$ would induce a correlation between $W_i$ and $\varepsilon_i$.\footnote{If $\zeta_i$ were observed, the analysis could condition on it, relaxing the requirement for $E[W_i \zeta_i] = 0$. However, $r(W_i, Z_i)$ is assumed unknown.} For simplicity, we focus on the just-identified case with a single endogenous regressor.\footnote{See \cite{ANDREWS2019} for an aggregation approach when multiple instruments are available.} To streamline the analysis, $W_i$ is assumed to be demeaned.

Identification of the APE requires additional moment conditions:
\begin{assumption}[IV Moment Conditions]
    \label{assump:IV_MC}
    Within the framework of Assumption \ref{assump:IV_DGP}, let the following moment conditions hold: 
    \begin{gather}
        E[W_i \zeta_i^m g_m(Z_i)] = 0 \label{eq:IV_MC1} \\
        E\left[ \frac{ W_i r(W_i, Z_i)^{p+1} }{(p+1)E[ W_i r(W_i, Z_i)]} - r(W_i, Z_i)^p \mid \zeta_i, Z_i \right] = 0  \label{eq:IV_MC2} 
    \end{gather}
    for $\{p, m\} \in \mathbb{N}$, $0 \leq p \leq M-1$ and $0 \leq m \leq M$.
\end{assumption}

Under these conditions, the IV estimand coincides with the average partial effect of $X_i$ on $Y_i$:

\begin{theorem}[Estimation of the APE with IV]
    \label{thm:IV}
    Under Assumptions \ref{assump:IV_DGP} and \ref{assump:IV_MC}, the IV estimand:
    \begin{align}
      \beta = \frac{E[Y_i W_i]}{E[X_i W_i]} 
    \end{align}
    is equivalent to the average partial effect of $X_i$ on $Y_i$:
    \begin{align}
      E_{X, Z, \varepsilon}\left[ \partial_{X_i} Y_i\right] 
    \end{align}
\end{theorem}

The validity of Theorem \ref{thm:IV} hinges on the moment conditions in Assumption \ref{assump:IV_MC}, which link the heterogeneity in $r(W_i, Z_i)$ and $g_m(Z_i)$ to the instrument $W_i$. By focusing on equations (\ref{eq:IV_MC1}) and (\ref{eq:IV_MC2}), the analysis avoids reliance on cross-moment conditions involving higher-order interactions among $\zeta_i$, $W_i$, and $Z_i$. Such cross-moment conditions would require detailed restrictions on the joint distribution of the covariates. Instead, the moment conditions avoid interdependencies, providing a tractable structure for identification. This approach not only simplifies the argument but also ensures that the assumptions remain interpretable and economically meaningful. By maintaining minimal restrictions, the results emphasize the central role of instrument exogeneity and the distributional properties enabling the identification of the APE.

\subsection*{Non linear IV exclusion restriction}
To understand the applicability of Theorem \ref{thm:IV}, we begin by analyzing when the moment condition 
$$
E[W_i \zeta_i^m g_m(Z_i)] = 0 \quad \text{for } m \in \{0, \dots, M\}
$$
is satisfied. Intuitively, this condition can be viewed as a generalization of the classical IV exclusion restriction $E[W_i \varepsilon_i] = 0$. While the latter requires that the instrument $W_i$ is uncorrelated with the structural error $\varepsilon_i$, the moment condition here explicitly accounts for the heterogeneous treatment effects framework introduced in Assumption \ref{assump:IV_DGP}.

To build intuition, note that $\zeta_i$ captures the endogenous component of the treatment variable $X_i$ while $g_m(Z_i)$ introduces heterogeneity in the treatment effects based on covariates $Z_i$. In essence, $E[W_i \zeta_i^m g_m(Z_i)] = 0$ requires that the variation in $W_i$ is orthogonal to all interactions between the treatment's endogenous component $\zeta_i$ (and its higher powers) and the heterogeneous effect coefficients $g_m(Z_i)$. 

One way to satisfy this condition is given by $E[W_i \zeta_i^m \mid Z_i]=0$, requiring that, conditional on observed covariates $Z_i$, the instrument $W_i$ is exogeneous to unobserved confounders $\zeta_i^m$. In the classic IV setting without heterogeneous treatment effects, only unconditional exogeneity is required. Strengthening this assumption, $E[W_i \mid \zeta_i, Z_i]=0$ also satisfies the condition by requiring mean independence of the instrument $W_i$ conditional on observables and unobservables. These are two possible ways in which equation (\ref{eq:IV_MC1}) is satisfied, alternaives exist and depend on the application at hand.

\subsection*{Moments of $W_i$ and $r(W_i, Z_i)$}
The second condition in Assumption \ref{assump:IV_MC}, ensuring that the IV estimand equals the APE, is given by
\begin{equation*}
    E\left[ \frac{ W_i r(W_i, Z_i)^{p+1} }{(p+1)E[ W_i r(W_i, Z_i)] } - r(W_i, Z_i)^p \mid \zeta_i, Z_i \right] = 0.
\end{equation*}
for $p \in \{1,...,M-1\}$. This equation is similar to Assumption \ref{assump_errors} in R-OLS, the main difference being the conditional expectation and the addition of $r(W_i, Z_i)$. $r(W_i, Z_i)$ describes how $X_i$ is generated based on the instrument and covariates and introduces strong assumptions on the $X$-DGP, if one wants to estimate the APE with IV based on Theorem \ref{thm:IV}. 

To see this, assume for the moment that $r(W_i, Z_i)=r(W_i)$, i.e. the X and Y-DGP are non-linear, but the X-DGP doesn't depend on controls $Z_i$. In this case IV estimates the APE if 
\begin{equation}
    \label{assump:IV_non_linear_W_non_linear_Y}
    \frac{E[W_ir(W_i)^{p+1} \mid \zeta_i, Z_i ]}{(p+1)E[W_ir(W_i)]} - E[r(W_i)^{p} \mid \zeta_i, Z_i ] = 0.
\end{equation}
holds. To my knowledge, condition (\ref{assump:IV_non_linear_W_non_linear_Y}) is not satisfied by any ordinary distribution, even if $r(W_i)=W_i^q$ with $q \in \{2, 3, ...\}$ and $p=1$. The more elaborate the funcitonal form of $r(W_i)$, the more complex the conditions on the distribution of $W_i$ become, casting doubt that any non-degenerate distribution can satisfy them. 

Alternatively, let the X-DGP be linear in the instrument, i.e. $r(W_i)=W_i$, then equation (\ref{assump:IV_non_linear_W_non_linear_Y}) simplifies to a moment condition for $W_i$ that is similar to Assumption \ref{assump_errors},
\begin{equation}
    \label{assump:IV_linear_W_non_linear_Y}
    \frac{E[W_i^{p+2} \mid \zeta_i, Z_i ]}{(p+1)E[W_i^{2}]} - E[W_i^{p} \mid \zeta_i, Z_i ] = 0
\end{equation}
the main difference being the conditioning on $(\zeta_i, Z_i)$ and usage of the instrument $W_i$ rather than the exogeneous error $\nu_i$. This condition becomes more tractable if one further assumes $E[W_i^{2} \mid \zeta_i, Z_i]=E[W_i^{2}]$, since then an instrument that linearly influences the treatment and satisfies the moments of the normal distribution up to order $M+1$, conditional on observables and unobservables, estimates the APE. However, these assumptions are strong and require a highly specific DGP, once again casting doubt that the conditions of Theorem \ref{thm:IV} are satisfied.

Lastly, let the X-DGP be non-linear in the instrument and the Y-DGP be linear in $X_i$, i.e. $M=1$ and consequently $p=0$. Then IV estimates the ATE for any distribution of $W_i$ that satisfies\footnote{Conditioning on $\zeta_i$ is not needed in the specific example, since $\zeta_i^p=1$ ensures $\zeta_i$ has no impact.}
\begin{equation}
    \label{assump:IV_non_linear_W_linear_Y}
    E[W_i r(W_i) \mid Z_i] = E[W_i r(W_i)].
\end{equation} Assuming further that the X-DGP is linear in $W_i$, i.e. $r(W_i)=W_i$ condition (\ref{assump:IV_non_linear_W_linear_Y})  simplifies to
\begin{equation}
    \label{assump:IV_linear_W_linear_Y}
    E[W_i^{2} \mid Z_i]=E[W_i^{2}]
\end{equation}
requiring only mean independence of the squared instrument for IV to estimate the ATE.

It remains to conclude that without covariates $Z_i$ and a X-DGP that is linear in the instrument, moment conditions similar to the R-OLS case ensure IV estimates the APE under heterogeneous treatment effects, albeit under strong mean independence assumptions. Similarly, without heterogeneous treatment effects, mean independence of the instrument ensures IV estimates the ATE. As soon as the X and Y-DGP are non-linear, much stricter moment conditions need to be satisfy to guarantee the IV estimand equals the APE.

Moving on to the more complex case of $r(W_i, Z_i)$, we evaluate the impact of $Z_i$ on identification of the APE with IV. Given the identification conditions without covariates, we focus on (1) a non-linear Y-DGP and a X-DGP linear in $W_i$ or (2) a linear Y-DGP and a X-DGP non-linear in $W_i$. 

Starting with the latter case, one arrives at the following condition\footnote{Since $M=1$, $p$ can only be zero meaning that conditioning on $\zeta_i$ is not necessary.}
\begin{equation}
    \label{assump:IV_non_linear_W_linear_Y_with_Z}
    E\left[  W_i r(W_i, Z_i) \mid Z_i \right] = E[ W_i r(W_i, Z_i)] 
\end{equation}
which even under additive separability of $W_i$ and $Z_i$ would require mean independence of $W_i r(W_i)$, $E[W_i r(Z_i)]=0$ and $E[W_i]=0$.\footnote{Additive separability: $r(W_i, Z_i) = r(W_i) + r(Z_i)$ and mean independence: $E[ W_i r(W_i) \mid Z_i] = E[ W_i r(W_i)]$.} More generally, a combination of independence of instrument $W_i$ and covariates $Z_i$ and functional form assumptions on $r(W_i, Z_i)$ is required. Just assuming linearity of the X-DGP in $W_i$ does not suffice.

Moving on to the case of a non-linear Y-DGP and a X-DGP linear in $W_i$, we assume one of the simplest X-DGPs: $r(W_i, Z_i) = W_i * r(Z_i)$. Under this assumption, condition (\ref{eq:IV_MC2}) becomes:
\begin{equation*}
    r(Z_i)^{p} \left( \frac{r(Z_i) E[W_i^{p+2} \mid \zeta_i, Z_i]}{(p+1)E[W_i^2 r(Z_i)]} - E[ W_i^p \mid \zeta_i, Z_i] \right) = 0
\end{equation*}
requiring a highly specific dependence structure between the instrument $W_i$ and controls $Z_i$ to be satisfied. Assuming additive separability makes the requirements on the dependence structure even stronger by incuring a binomial expansion of $(r(W_i) + r(Z_i))^{p+1}$. 

It remains to conclude, that Theorem \ref{thm:IV} establishes IV can recover the APE in contexts beyond those considered in \cite{Cuesta_Steins_Lemma}. However, the sufficiency conditions required for this result impose stringent restrictions on the data-generating process, raising significant doubts about their plausibility in empirical settings. A summary of these conditions can be found in Table \ref{tbl:IV_MC2}, serving as a reference for instances where Assumption \ref{assump:IV_MC} is satisfied.\footnote{\cite{Cuesta_Steins_Lemma} show that the IV-Estimator estimates the APE of $X$ on $Y$ if the treatment and instrument are jointly normally distributed. Joint normality, however, holds if and only if $X$ and $W$ are an affine transformation of iid, standard normal random variables. Therefore, \cite{Cuesta_Steins_Lemma}'s results cover the case of $M>1$ and linear $r(W_i)$ and demand joint normality, a more strict requirement than moment condition (\ref{assump:IV_linear_W_non_linear_Y}). Their results about non additively separable unobservables $\varepsilon_i$ are not covered in this paper.}
\begin{table}[]
    \caption{Overview in which scenario equation (\ref{eq:IV_MC2}) is satisfied}
    \resizebox{\textwidth}{!}{
    \begin{tabular}{c|c|c|c|c}
          & linear $r(W_i)$                                                                            & non-linear $r(W_i)$   & linear $r(W_i, Z_i)$  & non-linear $r(W_i, Z_i)$ \\ \hline
    $M=1$ & \begin{tabular}[c]{@{}c@{}}\cmark \\ \tiny (under moment condition \ref{assump:IV_linear_W_linear_Y})\end{tabular}      & \begin{tabular}[c]{@{}c@{}}(\cmark) \\ \tiny (under moment condition \ref{assump:IV_non_linear_W_linear_Y})\end{tabular}      & \begin{tabular}[c]{@{}c@{}}(\cmark) \\ \tiny (under moment condition \ref{assump:IV_non_linear_W_linear_Y_with_Z})\end{tabular} & \begin{tabular}[c]{@{}c@{}}(\cmark) \\ \tiny (under moment condition \ref{assump:IV_non_linear_W_linear_Y_with_Z})\end{tabular}    \\ \hline
    $M>1$ & \begin{tabular}[c]{@{}c@{}}(\cmark) \\ \tiny (under moment condition \ref{assump:IV_linear_W_non_linear_Y})\end{tabular}  & \begin{tabular}[c]{@{}c@{}}(\cmark) \\ \tiny (under moment condition \ref{assump:IV_non_linear_W_non_linear_Y})\end{tabular}  & \xmark & \xmark   
    \end{tabular}
    }
    \tablenote{The table shows under which assumptions for $M$ and $r(W_i, Z_i)$ the sufficient condition (\ref{eq:IV_MC2}) in Assumption \ref{assump:IV_MC} is satisfied. Linear $r(W_i)$ means $W_i$ enters $r(W_i)$ only with a power of one. Non-linear $r(W_i)$ allows for polynomials of $W_i$. Linear $r(W_i, Z_i)$ means $W_i$ enters the function only with a power of one, while no restrictions apply for $Z_i$. Non-linear $r(W_i)$ allows for polynomials of $W_i$ while $Z_i$ remains unrestricted.}
    \label{tbl:IV_MC2}
\end{table}

\begin{remark}[Controlling for $Z_i$ in IV]
    The above exposition assumes the IV estimation does not control for observed covariates $Z_i$. This is unrealistic given empirical practise. We can incorporate control variables via the Frisch-Waugh-Lovell theorem by residualising the instrument $W_i$ wrt. $Z_i$ leading to the following estimator:
    $$
        \frac{E[Y_i \tilde{W}_i]}{E[X_i \tilde{W}_i]}
    $$
    where $\tilde{W}_i$ symbolises residualising $W_i$ wrt. to $Z_i$. In this setup, Assumption \ref{assump:IV_MC} becomes:
    \begin{gather*}
        E[\tilde{W}_i \zeta_i^m g_m(Z_i)] = 0 \\
        E\left[ \frac{ \tilde{W}_i r(W_i, Z_i)^{p+1} }{(p+1)E[ \tilde{W}_i r(W_i, Z_i)]} - r(W_i, Z_i)^p \mid \zeta_i, Z_i \right] = 0  
    \end{gather*}
    for $\{p, m\} \in \mathbb{N}$, $0 \leq p \leq M-1$ and $0 \leq m \leq M$. This leads to a relaxation of the mean independence assumptions proposed above. Assumption \ref{assump:IV_MC} remains strong, yet more plausible in empiricial applications.
\end{remark}

\begin{proof}[Proof of Theorem \ref{thm:IV}]
    Under Assumption \ref{assump:IV_DGP}, the IV estimand can be written as:
    \begin{align*}
        \frac{E[W_i Y_i]}{E[W_i X_i] } &= \frac{1}{E[ W_i X_i] } E\left[W_i (\sum_{m=0}^M X_{i}^mg_m(Z_i) + \varepsilon_i) \right] \\
        \shortintertext{Exogeneity of the instrument $W_i$ and definition of the X-DGP:}
        &= \frac{1}{E[ W_i r(W_i, Z_i)] + E[W_i \zeta_i]} E\left[W_i \sum_{m=0}^M (r\left(W_i, Z_i\right)+\zeta_{i})^m g_m(Z_i) \right] \\
        \shortintertext{Binomial theorem:}
        &= \frac{1}{E[ W_i r(W_i, Z_i)] } \sum_{m=0}^M \sum_{p=0}^m {m\choose p} E\left[W_i r\left(W_i, Z_i\right)^{p}\zeta_{i}^{m-p} g_m(Z_i) \right] \\
        \shortintertext{Manipulation of summation and binomial coefficient:}
        &= \sum_{m=0}^M \left[ \frac{E\left[W_i \zeta_i^{m} g_m(Z_i) \right]}{E[ W_i r(W_i, Z_i)] } + \sum_{p=1}^{m} {m\choose p} \frac{E\left[W_i r\left(W_i, Z_i\right)^{p}\zeta_{i}^{m-p} g_m(Z_i) \right]}{E[ W_i r(W_i, Z_i)] } \right] \\
        &= \sum_{m=0}^M \left[ \frac{E\left[W_i \zeta_i^{m} g_m(Z_i) \right]}{E[ W_i r(W_i, Z_i)] } + \sum_{p=0}^{m-1} {m\choose p+1} \frac{E\left[W_i r\left(W_i, Z_i\right)^{p+1}\zeta_{i}^{m-p-1} g_m(Z_i) \right]}{E[ W_i r(W_i, Z_i)] } \right] \\
        &= \sum_{m=0}^M \left[ \underbrace{\frac{E\left[W_i \zeta_i^{m} g_m(Z_i) \right]}{E[ W_i r(W_i, Z_i)] }}_{=0, \text{ by A\ref{assump:IV_MC}}} + \sum_{p=0}^{m-1} {m-1 \choose p} m \underbrace{\frac{E\left[W_i r\left(W_i, Z_i\right)^{p+1}\zeta_{i}^{m-p-1} g_m(Z_i) \right]}{(p+1) E[ W_i r(W_i, Z_i)] }}_{=E\left[ r\left(W_i, Z_i\right)^{p}\zeta_{i}^{m-p-1} g_m(Z_i) \right], \text{ by A\ref{assump:IV_MC}}} \right] \\
        &= \sum_{m=0}^M \sum_{p=0}^{m-1} {m-1 \choose p} m E\left[ r\left(W_i, Z_i\right)^{p}\zeta_{i}^{m-p-1} g_m(Z_i) \right] \\
        \shortintertext{Reverse binomial theorem:}
        &= E\left[ \sum_{m=0}^M m (r\left(W_i, Z_i\right)+\zeta_{i})^{m-1}g_m(Z_i) \right] \\
        \shortintertext{Definition of the X-DGP:}
        &= E\left[ \sum_{m=0}^M m X_{i}^{m-1}g_m(Z_i) \right] \\
        \shortintertext{Derivative of the Y-DGP:}
        &= E\left[ \frac{\partial Y_i}{\partial X_i} \right]
    \end{align*}
    where the first stes are virtually identical to the steps used in the proof of Theorem \ref{thm:main} with Assumption \ref{assump:IV_DGP} rather than Assumption \ref{assump_DGP} being used. The main difference lies in the binomial expansion being switched around, avoiding the heroic assumption that $E\left[W_i r\left(W_i, Z_i\right)^{m} g_m(Z_i) \right]=0$. When applying Assumption \ref{assump:IV_MC} we used:
    \begin{align*}
        0 &= \frac{E\left[W_i r\left(W_i, Z_i\right)^{p+1}\zeta_{i}^{m-p-1} g_m(Z_i) \right]}{(p+1) E[ W_i r(W_i, Z_i)] } - E\left[ r\left(W_i, Z_i\right)^{p}\zeta_{i}^{m-p-1} g_m(Z_i) \right] \\
        &= E\left[ \frac{W_i r\left(W_i, Z_i\right)^{p+1}\zeta_{i}^{m-p-1} g_m(Z_i) }{(p+1) E[ W_i r(W_i, Z_i)] } - r\left(W_i, Z_i\right)^{p}\zeta_{i}^{m-p-1} g_m(Z_i) \right] \\
        &= E\left[ \zeta_{i}^{m-p-1} g_m(Z_i) \underbrace{E\left[ \frac{W_i r\left(W_i, Z_i\right)^{p+1}}{(p+1) E[ W_i r(W_i, Z_i)] } - r\left(W_i, Z_i\right)^{p} \mid \zeta_{i}, Z_i \right]}_{=0, \text{ by A\ref{assump:IV_MC}}} \right],
    \end{align*}
    where the last step uses the law of iterated expectation to factor $\zeta_i$ and $g_m(Z_i)$ out via the conditional expectation. This shows the IV estimand is equivalent to the APE, under the DGP in Assumption \ref{assump:IV_DGP} and moment conditions in Assumption \ref{assump:IV_MC}, completing the proof.
\end{proof}

\end{document}